\pdfoutput=1
\documentclass[11pt]{article}

\usepackage[T1]{fontenc}
\usepackage[utf8]{inputenc}
\usepackage{lmodern}
\usepackage[a4paper,margin=1in]{geometry}
\usepackage{microtype}
\usepackage{amsmath,amssymb,amsthm,mathtools}
\usepackage{bbm}
\usepackage{enumitem}
\usepackage{booktabs}
\usepackage{array}
\usepackage{graphicx}
\usepackage{placeins}
\usepackage{xcolor}
\usepackage{hyperref}
\hypersetup{colorlinks=true,linkcolor=blue!45!black,citecolor=blue!45!black,urlcolor=blue!45!black}

\theoremstyle{definition}
\newtheorem{definition}{Definition}[section]
\newtheorem{example}[definition]{Example}

\theoremstyle{plain}
\newtheorem{theorem}[definition]{Theorem}
\newtheorem{lemma}[definition]{Lemma}
\newtheorem{proposition}[definition]{Proposition}
\newtheorem{corollary}[definition]{Corollary}
\theoremstyle{remark}
\newtheorem{remark}[definition]{Remark}

\DeclareMathOperator{\TV}{TV}

\DeclareMathOperator{\Ker}{Ker}

\DeclareMathOperator{\argmin}{arg\,min}
\DeclareMathOperator{\poly}{poly}

\DeclareMathOperator{\rank}{rank}
\DeclareMathOperator{\Tr}{Tr}

\newcommand{\E}{\mathbb{E}}
\newcommand{\Prb}{\mathbb{P}}
\newcommand{\R}{\mathbb{R}}

\newcommand{\cA}{\mathcal{A}}

\newcommand{\cC}{\mathcal{C}}

\newcommand{\cF}{\mathcal{F}}

\newcommand{\cH}{\mathcal{H}}

\newcommand{\cM}{\mathcal{M}}
\newcommand{\cO}{\mathcal{O}}

\newcommand{\cT}{\mathcal{T}}

\newcommand{\cW}{\mathcal{W}}

\newcommand{\DeltaS}{\Delta}
\newcommand{\eps}{\varepsilon}
\newcommand{\ind}{\mathbbm{1}}
\newcommand{\eqobs}{\equiv_{\Pi}}
\newcommand{\eqep}{\equiv_{\Pi,\eps}}

\newcommand{\Adv}{\operatorname{Adv}}
\newcommand{\Real}{\mathsf{Real}}
\newcommand{\Ideal}{\mathsf{Ideal}}

\title{Observer-Quotient Security:\ Composable Leakage Bounds for Hidden State Continuations}

\author{
Faruk Alpay\thanks{Correspondence: \texttt{alpay@lightcap.ai}.} \quad Levent Sar{\i}o\u{g}lu\\
Department of Computer Engineering, Bah\c{c}e\c{s}ehir University, Istanbul, T\"urkiye\\
\texttt{\{faruk.alpay, levent.sarioglu\}@bahcesehir.edu.tr}
}

\date{July 2026}

\begin{document}
\maketitle

\begin{abstract}
Observer-quotient security studies cryptographic executions through the distributions visible to an explicit observer class. The internal state of a protocol, implementation, or controlled plant may evolve inside an observer quotient while the public transcript remains unchanged; security then depends on whether later continuations, richer side channels, repeated sessions, or post-processing can turn that hidden evolution into distinguishing advantage. This paper formalizes interactive observer-quotient security games with session identifiers, oracle forwarding, adaptive observer schedulers, nonuniform environments, simulator/adaptors, ideal quotient functionalities, and scalar boundary functionals that upper-bound leakage under time evolution. The main real/ideal theorem is proved through named interface lemmas, an explicit wrapper construction, and a fully indexed hybrid sequence whose defects are observer, kernel, post-processing, simulator, statistical, and residual terms. The computational instantiations include IND-CPA encryption with timing leakage, deterministic encryption under alphabet-size and smooth max-information leakage, and nonce-refreshing authenticated-encryption state machines. The same quotient geometry yields the optimization-control layer: hidden continuations are observability kernels, observer refinement is sensor redesign, stochastic dissipativity gives residual leakage floors, and cost-constrained observer synthesis converts sensor choices into concrete reductions in bounded distinguishing advantage.
\end{abstract}
\section{Introduction}

Security assertions are statements about observations. Two executions may differ internally while remaining indistinguishable to a transcript observer, and that indistinguishability may fail after the observer class is refined by timing, cache, power, electromagnetic, or profiled leakage. The relevant state space is therefore not only the implementation state space $X$, but the quotient of $X$ induced by the observer family used in the security game.

Let $\Pi$ be a family of observers. States $x,x'\in X$ are equivalent when every observer in $\Pi$ gives the same distribution. A transition from $x$ to $x'$ inside one equivalence class is a hidden continuation: the computation has changed, while the admitted view has not. Nonces may refresh, counters may advance, randomness may be consumed, masked shares may rotate, and microarchitectural state may move without altering the public transcript. The security question is whether an adaptive environment can later expose that continuation through a scheduled observer, a side channel, a further session, or an efficient post-processing of the transcript.

This paper develops a security-game calculus for that question. The primitive object is an interactive experiment with sessions, oracle calls, observer choices, continuation commands, leakage labels, and an ideal quotient functionality. The real/ideal theorem treats hidden continuations as substitutable resources precisely when they remain compatible with the quotient interface and when the residual leakage functional stays below the corresponding advantage bound. The proof is organized by named interfaces and named hybrids rather than by an implicit appeal to observational equivalence.

The resulting framework also has a control interpretation. An observer is a sensor, an observer quotient is an output-equivalence partition, and a hidden continuation is an element of an observability kernel. Sensor redesign refines the quotient; persistent noise or implementation drift appears as an input-to-state residual; and a cost-constrained observer-design problem trades measurement cost against reduction in bounded distinguishing advantage. This control layer is used as a proof language for leakage bounds, while the primary object remains the cryptographic interactive game.

\paragraph{Contributions.}
The paper makes four contributions. First, it defines observer-quotient security games and the associated pseudometric, quotient dynamics, and hidden-continuation structure. Second, it gives a compositional real/ideal theorem for adaptive environments with oracle forwarding, scheduled observers, simulator/adaptors, abort synchronization, and post-processing closure. Third, it turns the abstract objects into audit-grade cryptographic reductions, including an IND-CPA-with-timing theorem with explicit query, profiling, scheduler, and noise parameters. Fourth, it connects observer refinement to side-channel validation and sensor design through finite-state tests, multiple-comparison envelopes, LTI observability kernels, KKT conditions, and residual-floor-to-advantage conversions.

\section{Related Work}

Indistinguishability is a foundational mode of modern cryptography, from semantic security to simulation-based and universally composable definitions \cite{Yao1982,GoldwasserMicali1984,Canetti2001}. In these settings, security is a statement about the inability of a class of adversaries to separate two experiments. Process calculi and bisimulation use a similar idea: systems can be behaviorally equivalent even if their internal structures differ \cite{Park1981,Milner1989}. Noninterference and information-flow security also reason about which variations of hidden state can be observed at public outputs \cite{GoguenMeseguer1982,SabelfeldSands2009,ClarksonSchneider2010}.

The development here isolates the algebraic object that they share. The observer class $\Pi$ induces a quotient on the internal state space. Security then becomes a fixed-point or invariance statement on that quotient, not on raw states. This formulation makes side-channel omissions and telemetry choices explicit: changing $\Pi$ changes the quotient and can destroy a proof.

There is also a natural relation to categorical cryptography and composable security, where protocols and resources are treated as morphisms and composition is central \cite{BroadbentKarvonen2022}. Our treatment is less categorical and more operational: the primitive objects are Markov kernels, transcript distributions, adversary advantage, and optimization over observer channels. The categorical viewpoint is compatible, but not required.

The optimization and control side is equally important. Classical observability asks whether a state can be reconstructed from output trajectories \cite{Kalman1960,Sontag1998}. If it cannot, then the unobservable subspace is exactly a space of hidden continuations. Secure control and resilient estimation study adversarial actions that remain hidden from available sensors; privacy and information-flow optimization study which measurements are released under leakage constraints \cite{CoverThomas2006,Alvim2014}. Physical-layer security and communications security also contain many problems where security design is expressed as an optimization problem over channels, powers, rates, or sensors \cite{WangBaiZhaoHan2019}. Our contribution is to place these under the same quotient semantics as cryptographic indistinguishability.

\section{Mathematical Preliminaries}

\subsection{Spaces, channels, and total variation}

For a finite set $S$, write $\DeltaS(S)$ for the probability simplex over $S$. A randomized observer from $X$ to $Y$ is a Markov kernel $\pi:X\to \DeltaS(Y)$. When $X$ and $Y$ are finite, $\pi(y\mid x)$ denotes the probability of output $y$ on input state $x$. For probability distributions $P,Q\in\DeltaS(Y)$, the total variation distance is
\[
\TV(P,Q)=\frac12\sum_{y\in Y}|P(y)-Q(y)|.
\]
For general measurable spaces, the same definition is $\sup_{A}|P(A)-Q(A)|$. The finite case covers the concrete security-game statements below; the measurable case follows by standard kernel arguments.

\subsection{Observer classes}

An observer class $\Pi$ is a set of admissible observation channels. In cryptography, $\Pi$ may be all polynomial-time adversarial views or all distinguishers with a given resource bound. In side-channel analysis, $\Pi$ may be a family of timing, cache, power, electromagnetic, or transcript extractors. In control, $\Pi$ may be a set of sensors or output maps.

\begin{definition}[Observer pseudometric]
Let $X$ be a state space and let $\Pi$ be a class of channels from $X$ to finite output spaces. Define
\[
    d_{\Pi}(x,x')=\sup_{\pi\in\Pi}\TV(\pi(\cdot\mid x),\pi(\cdot\mid x')).
\]
For $\eps\ge 0$, write $x\eqep x'$ when $d_{\Pi}(x,x')\le \eps$. For $\eps=0$, write $x\eqobs x'$.
\end{definition}

The relation $\eqobs$ is an equivalence relation. For $\eps>0$, the relation $\eqep$ is reflexive and symmetric but need not be transitive without changing the tolerance. The function $d_{\Pi}$ is a pseudometric: distinct states can have zero distance if no observer in $\Pi$ separates them.

\begin{definition}[Observer quotient]
The observer quotient of $X$ by $\Pi$ is the set
\[
    X/\Pi = X/{\eqobs}.
\]
The quotient map is denoted $q_{\Pi}:X\to X/\Pi$.
\end{definition}

\begin{remark}
The quotient depends on the adversarial model. Enlarging $\Pi$ refines the quotient. Shrinking $\Pi$ coarsens it. A proof that is true for a small observer class may be false after a single extra side channel is admitted.
\end{remark}

\subsection{Fixed points and hidden continuations}

Let $T:X\to X$ be a deterministic transition or, more generally, let $K:X\to\DeltaS(X)$ be a transition kernel. In the deterministic case, a raw fixed point satisfies $T(x)=x$. An observational fixed point satisfies $T(x)\eqobs x$.

\begin{definition}[Hidden continuation]
For deterministic $T:X\to X$, a state $x$ has a hidden continuation under $\Pi$ when
\[
    T(x)\ne x \quad\text{and}\quad T(x)\eqobs x.
\]
For a kernel $K$, $x$ has an $\eps$-hidden continuation when the output distribution generated by $K(x)$ is within $\eps$ of that generated by the point mass at $x$ for every observer in $\Pi$.
\end{definition}

Hidden continuation is the mathematical name for a familiar security phenomenon. A protocol can evolve internally while its transcript remains unchanged. A plant can move in an unobservable subspace. A program can update a secret-dependent cache state while its functional output remains the same. The security question is whether future observations can reveal that hidden movement.

\section{Observer-Quotient Security Games}

\subsection{Experiments and transcript observers}

A security experiment is a tuple
\[
    \mathsf{Exp}_b=(X,x_0^b,K_b,\Pi,\cA,T),\qquad b\in\{0,1\},
\]
where $X$ is the internal state space, $x_0^b$ is the initial state in world $b$, $K_b$ is a transition kernel, $\Pi$ is the admissible observer class, $\cA$ is an adversary class, and $T$ is a horizon. At each time $t$, the state evolves as $X_{t+1}\sim K_b(X_t)$ and an observer channel $\pi_t\in\Pi$ produces $Z_t\sim \pi_t(\cdot\mid X_t)$. The choice of $\pi_t$ may be fixed, adaptive, or controlled by an adversary subject to a resource bound.

The transcript is $Z_{0:T}=(Z_0,\ldots,Z_T)$. For an adversary $A\in\cA$ outputting a bit, define the advantage
\[
\operatorname{Adv}_{\Pi,T}^{\cA}(\mathsf{Exp}_0,\mathsf{Exp}_1)
=\sup_{A\in\cA}\left|\Prb[A(Z_{0:T})=1\mid b=0]-\Prb[A(Z_{0:T})=1\mid b=1]\right|.
\]
When $\cA$ is the class of all tests, the advantage is exactly total variation distance between transcript distributions. When $\cA$ is polynomial-time, this is the usual computational indistinguishability style of definition.

\begin{definition}[Observer-quotient security]
The pair $(\mathsf{Exp}_0,\mathsf{Exp}_1)$ is $(\Pi,\cA,T,\eps)$-secure when
\[
    \operatorname{Adv}_{\Pi,T}^{\cA}(\mathsf{Exp}_0,\mathsf{Exp}_1)\le \eps.
\]
It is observer-quotient secure when the above holds for the stated observer and adversary class.
\end{definition}

\begin{example}[IND-style encryption as observer-quotient security]
In a left-or-right encryption experiment, the hidden bit $b$ selects which plaintext is encrypted. The internal state includes the key, randomness, nonce schedule, and message. The observer sees the ciphertext and any admitted side-channel outputs. IND security is recovered by letting $\Pi$ output the adversary-visible transcript and letting $\cA$ be the chosen computational adversary class.
\end{example}

\begin{example}[Sensor security]
In a controlled plant, world $b=0$ may be the nominal trajectory and world $b=1$ may be an attacked trajectory. The observer class is the set of available sensors and telemetry processors. Security means that an adversary cannot force a harmful internal deviation while keeping the sensor transcript inside the same observer quotient.
\end{example}

\subsection{Threat-model discipline}

The notation intentionally forces three choices to be stated:
\begin{enumerate}[leftmargin=2em]
    \item \textbf{Which observers are included?} Transcript-only security and transcript-plus-timing security are different definitions.
    \item \textbf{Which adversaries are included?} Unbounded distinguishers, polynomial-time distinguishers, adaptive controllers, and physical side-channel observers induce different games.
    \item \textbf{Which horizon is used?} A one-step proof may not control future leakage.
\end{enumerate}
The three parameters determine the security statement: changing any of them changes the experiment.

\section{Computational Observer-Quotient Games}

The statistical quotient definitions become cryptographic once the machines, advice, oracle access, and resource bounds are part of the experiment. Throughout this section $\lambda$ is the security parameter, all encodings have polynomial length, and the environment, observers, schedulers, simulators, and distinguishers are either uniform PPT machines or nonuniform polynomial-size circuits with advice.

\begin{definition}[Computational OQSG]
A computational observer-quotient security game is a family
\[
    G_\lambda=(X_\lambda,\Sigma_\lambda,K^0_\lambda,K^1_\lambda,\Pi_\lambda,\cA_\lambda,\cO_\lambda,T_\lambda)
\]
with encoded state space $X_\lambda$, transcript alphabet $\Sigma_\lambda$, transition kernels $K^b_\lambda$ sampled by PPT algorithms with access to the oracle family $\cO_\lambda$, an observer family $\Pi_\lambda$ of PPT transducers, an adversary class $\cA_\lambda$, and horizon $T_\lambda=\poly(\lambda)$. A nonuniform adversary receives advice $a_\lambda$ and interacts with the experiment through the declared interface.
\end{definition}

For an adaptive adversary $A$ and observer scheduler $S$, write $\mathsf{View}^{b,A,S}_{0:T}$ for the transcript distribution in world $b$. The concrete advantage is
\[
\Adv^{A,S}_{G}(\lambda,T)=\left|\Pr[A(1^\lambda,a_\lambda,\mathsf{View}^{0,A,S}_{0:T})=1]-\Pr[A(1^\lambda,a_\lambda,\mathsf{View}^{1,A,S}_{0:T})=1]\right|.
\]
A family is computationally observer-quotient secure for a resource class when the supremum of this quantity over that class is negligible; in concrete statements the dependence on runtime, query count, profiling samples, noise, and horizon is retained.

\begin{definition}[Computational observer pseudometric]
For encoded states $x,x'\in X_\lambda$ and circuit size $q$, define
\[
 d^{\mathsf{ppt}}_{\Pi_\lambda,q,t}(x,x') =
 \sup_{\substack{\pi\in\Pi_\lambda\text{ scheduled by time }t\\D\text{ circuit size }q}}
 \left|\Pr[D(\pi(x))=1]-\Pr[D(\pi(x'))=1]\right|.
\]
The relation $x\equiv^{\mathsf{ppt}}_{\Pi,\epsilon}x'$ means this value is at most $\epsilon$ for the stated resources.
\end{definition}

\begin{definition}[Boundary advantage functional]
Let $B_\Pi:X_\lambda\to\mathbb R_+$ be a separating quotient boundary gap and let $Z_t$ be an adapted scalar surrogate. A boundary advantage functional is a family $(\rho_t,\eta_t)$ such that every admissible scheduler and distinguisher satisfies
\[
\Adv_t\le \rho_t(\E Z_t)+\eta_t,
\]
where $\rho_t$ is nondecreasing and $\eta_t$ is a negligible term or a concrete statistical residual. When $Z_t$ obeys
\[
\E[Z_{t+1}\mid\cF_t]\le Z_t-a_t\phi(Z_t)+\chi_t,
\]
the same scalar recurrence controls the residual distinguishing floor.
\end{definition}

\begin{lemma}[Averaging for indexed hybrids]
Let $H_0,\ldots,H_N$ be experiments and let $D$ be a fixed distinguisher. If
\[
\left|\Pr[D(H_0)=1]-\Pr[D(H_N)=1]\right|>\Delta,
\]
then some adjacent pair satisfies
\[
\left|\Pr[D(H_j)=1]-\Pr[D(H_{j+1})=1]\right|>\Delta/N.
\]
The same bound holds after conditioning on any prefix distribution that is sampled identically in all adjacent replacements.
\end{lemma}

\begin{proof}
The difference telescopes across adjacent experiments. The triangle inequality gives the unconditional statement, and applying the same argument inside each prefix-conditioned experiment gives the conditional statement after averaging over the prefix.
\end{proof}

\section{Interactive Interfaces, Morphisms, and Emulation}

\begin{definition}[Interactive OQSG interface]
An interactive observer-quotient game is a tuple
\[
\mathfrak G=(\mathsf{Init},\mathsf{Invoke},\mathsf{Continue},\mathsf{Observe},\mathsf{Leak},\mathsf{Post},\mathsf{Abort})
\]
with state space $X_\lambda$, session-identifier space $\mathsf{SID}_\lambda$, observer class $\Pi_\lambda$, oracle family $\cO_\lambda$, and transcript alphabet $\Sigma_\lambda$. The command $\mathsf{Init}(1^\lambda)$ samples keys, public parameters, implementation randomness, and an initial state. The command $\mathsf{Invoke}(\mathsf{sid},u)$ creates or resumes a session. The command $\mathsf{Continue}(\mathsf{sid},a)$ applies an admissible continuation action. The command $\mathsf{Observe}(\mathsf{sid},\pi)$ returns a sample from the scheduled observer $\pi\in\Pi_\lambda$. The command $\mathsf{Leak}(\mathsf{sid},j)$ returns an admitted leakage sample. The command $\mathsf{Post}(g)$ applies an efficient post-processing map to the public view. The command $\mathsf{Abort}$ records the session abort bit. No command returns a raw representative.
\end{definition}

\begin{definition}[Environment, scheduler, context, and simulator]
A nonuniform environment $Z=(Z_\lambda,a_\lambda)$ is a polynomial-size interactive machine. At each activation it may choose session identifiers, issue oracle calls, schedule continuation actions, select an observer, and request post-processing. An observer scheduler $S$ is the subroutine that maps the current public view to $(\mathsf{sid}_t,\pi_t,g_t)$ with $\pi_t\in\Pi_\lambda$. A context $\cC[\cdot]$ is any polynomial-size wrapper built from sequential composition, parallel composition, oracle forwarding, observer scheduling, and post-processing. A simulator/adaptor $\mathcal S$ is a PPT machine that maps the real interface to the ideal quotient interface, including session identifiers, oracle replies, public randomness, leakage labels, and abort behavior.
\end{definition}

\begin{definition}[Ideal quotient functionality]
The ideal quotient functionality $\mathcal F_{\Pi,B}$ stores $(q_\Pi(x),\mathsf{sid},\mathsf{tr},z)$: the quotient class, session identifier, public transcript, and scalar boundary value. Its transition is the quotient pushforward of the compatible kernel. When a representative is needed for sampling, the simulator samples it conditionally on the quotient class and the public transcript. Leakage labels are sampled from the declared ideal leakage law and remain consistent with the quotient class.
\end{definition}

\begin{lemma}[Scheduling interface]
For any admissible environment and context, the sequence $(\mathsf{sid}_t,\pi_t,g_t)$ selected at time $t$ is measurable with respect to the public prefix $\mathsf{tr}_{<t}$ and the advice $a_\lambda$. Consequently a transition/output replacement at time $t$ may condition on the entire prefix without changing the distribution of the scheduled observer for that step.
\end{lemma}

\begin{proof}
Admissibility gives that the context and scheduler read only the public prefix, oracle replies already returned by the interface, and the fixed advice. The raw state and raw representative are absent from the interface. Conditioning on the public prefix therefore fixes the scheduler input and hence the scheduled triple for the next hop.
\end{proof}

\begin{lemma}[Representative sampling]
Let $Q$ be a quotient class and let $\mathsf{tr}$ be a public transcript prefix with positive probability in the real execution. If the simulator samples a representative from the conditional law $\mathsf{Law}(X_t\mid q_\Pi(X_t)=Q,\mathsf{tr}_{<t}=\mathsf{tr})$, then every scheduled observer in $\Pi$ has the same conditional output law as the real execution, up to the declared observer defect $\delta_{\mathrm{obs},t}$.
\end{lemma}

\begin{proof}
The conditional law preserves the quotient class and the public prefix. For any scheduled $\pi_t$, quotient equivalence gives equality of observer distributions inside the class in the exact case; the approximate case contributes $\delta_{\mathrm{obs},t}$ by definition of the observer defect.
\end{proof}

\begin{lemma}[Abort synchronization]
Let $A_t$ and $A_t^{\mathsf{id}}$ be the real and ideal abort bits at time $t$. If
\[
\TV\bigl(\mathsf{Law}(A_t\mid\mathsf{tr}_{<t}),\mathsf{Law}(A_t^{\mathsf{id}}\mid\mathsf{tr}_{<t})\bigr)\le \delta_{\mathrm{abort},t},
\]
then replacing real abort behavior by ideal abort behavior over a horizon $T$ changes the environment's output distribution by at most $\sum_{t<T}\delta_{\mathrm{abort},t}$.
\end{lemma}

\begin{proof}
Couple the abort bit at each step conditionally on the prefix. The probability that any coupling fails is bounded by the sum of the conditional total variations, and the environment's final output can change only on that event.
\end{proof}

\begin{lemma}[Leakage-label consistency]
Assume every leakage label $\ell_t$ is a deterministic or randomized function of $(q_\Pi(X_t),\mathsf{sid}_t,\mathsf{tr}_{<t})$ in the ideal interface, with conditional defect $\delta_{\mathrm{label},t}$ from the real label distribution. Then leakage-label replacement over the horizon changes any admissible environment's output law by at most $\sum_{t<T}\delta_{\mathrm{label},t}$.
\end{lemma}

\begin{proof}
This is the same prefix coupling argument as the abort lemma, applied to the leakage-label alphabet. The scheduled post-processing map is prefix-measurable by the scheduling lemma, so label replacement is performed before the next activation without additional loss.
\end{proof}

\begin{definition}[Post-processing closure]
An observer family $\Pi$ is closed under admissible post-processing with defect $\delta_{\mathrm{post},t}$ if for every scheduled $\pi_t\in\Pi$ and every PPT transducer $g_t$ selected by the environment, the composed channel $g_t\circ\pi_t$ is represented in $\Pi$ at step $t$ up to total-variation defect $\delta_{\mathrm{post},t}$. It is closed under adaptive scheduling when the prefix-measurable selection of $\pi_t$ can be included in the next-hop distinguisher.
\end{definition}

\begin{definition}[OQSG morphism with defect vector]
For games $G=(X,K,\Pi,B,\cA)$ and $G'=(X',K',\Pi',B',\cA')$, an OQSG morphism $M:G\to G'$ is a tuple
\[
    M=(f,m,\alpha,\beta,s)
\]
where $f:X\to X'$ maps states, $m$ maps kernels, $\alpha:\Pi'\to\Pi$ pulls observers back, $\beta:\cA'\to\cA$ pulls distinguishers back, and $s:\mathbb R_+\to\mathbb R_+$ maps boundary values. Its step-$t$ defect vector is
\[
\mathsf{def}_t(M)=(\delta_{\mathrm{obs},t},\delta_{K,t},\delta_{B,t},\delta_{\mathrm{post},t}).
\]
These terms bound, respectively, observer pullback error, quotient-kernel pushforward error, boundary-value distortion, and post-processing closure loss.
\end{definition}

\begin{proposition}[Composition law for OQSG morphisms]
Let $M_1:G_0\to G_1$ and $M_2:G_1\to G_2$ be morphisms.
\begin{enumerate}[leftmargin=2em]
    \item Sequential composition is obtained by composing $f,m,s$ covariantly and $\alpha,\beta$ contravariantly. Its observer and kernel defects satisfy
    \[
    \delta_{\mathrm{obs}}(M_2\circ M_1)\le\delta_{\mathrm{obs}}(M_1)+\delta_{\mathrm{obs}}(M_2),\quad
    \delta_K(M_2\circ M_1)\le\delta_K(M_1)+\delta_K(M_2).
    \]
    If $s_2$ is $L_{s_2}$-Lipschitz on the relevant range, then
    \[
    \delta_B(M_2\circ M_1)\le L_{s_2}\delta_B(M_1)+\delta_B(M_2).
    \]
    \item Parallel composition of independent components has product total-variation defect
    \[
    \delta^{\parallel}_{\mathrm{obs}}=1-\prod_i(1-\delta^i_{\mathrm{obs}})\le\sum_i\delta^i_{\mathrm{obs}},
    \]
    and a weighted scalar functional $Z^{\parallel}=\sum_i\omega_iZ^i$ with residual $\chi^{\parallel}=\sum_i\omega_i\chi^i$.
    \item If the observer at step $t$ is selected adaptively from the prefix, the same inequalities hold conditionally on the prefix and the horizon defect is the predictable sum of the step defects.
\end{enumerate}
\end{proposition}

\begin{proof}
Sequential composition follows by inserting the intermediate observer distribution and applying the triangle inequality. The boundary inequality follows from $s_2(s_1(B)+\delta_B(M_1))+\delta_B(M_2)$. The parallel statement uses $\TV(\otimes_iP_i,\otimes_iQ_i)\le1-\prod_i(1-\TV(P_i,Q_i))$. Adaptive scheduling is handled by conditioning on the prefix, where the selected observer is fixed.
\end{proof}

\begin{lemma}[Wrapper construction]
Fix a context $\cC$, environment $Z$ with advice $a_\lambda$, scheduler $S$, and distinguisher output bit. For any adjacent hybrid hop at time $t$, there is an admissible distinguisher $W^{\cC,Z,S}_t$ for that hop with state variables
\[
(\mathsf{sid}_{\le t},\mathsf{tr}_{\le t},\mathsf{oracles}_{\le t},\mathsf{leak}_{\le t},\mathsf{abort}_{\le t},a_\lambda,\mathsf{st}_Z,\mathsf{st}_S,\mathsf{st}_{\cC}).
\]
It samples the prefix exactly as in the global experiment, submits the challenged transition/output pair to the hop oracle, resumes $\cC[\cdot]$, and outputs $Z$'s final bit. Its runtime is
\[
\operatorname{time}(W_t)\le \operatorname{time}(Z)+\operatorname{time}(S)+\operatorname{time}(\cC)+\poly(T,\lambda),
\]
and its oracle count is the oracle count of the global execution plus the single challenge call for the hop.
\end{lemma}

\begin{proof}
The wrapper stores the complete public prefix and the internal states of the public machines. The scheduling lemma fixes the next observer from this prefix. The wrapper embeds the adjacent-hop challenge at the selected command and forwards all other oracle calls honestly. After the challenge output is returned, the wrapper continues the same public machines and uses the environment's final bit. All stored data have polynomial length and all machines are polynomial-time, giving the displayed overhead.
\end{proof}

\begin{definition}[Real/ideal OQSG emulation]
A real game $G$ emulates an ideal quotient functionality $\mathcal F_{\Pi,B}$ with error $\epsilon(\lambda,T)$ if for every polynomial-time context $\cC$, every nonuniform environment $Z$, and every adaptive observer scheduler $S$ admissible for $\Pi$, there exists a PPT simulator/adaptor $\mathcal S$ such that
\[
\left|
\Pr[Z^{\cC[G],S}(1^\lambda,a_\lambda)=1]
-
\Pr[Z^{\cC[\mathcal S^{\mathcal F_{\Pi,B}}],S}(1^\lambda,a_\lambda)=1]
\right|
\le \epsilon(\lambda,T).
\]
The quantifier order is
\[
\forall \cC\;\forall Z\;\forall S\;\exists \mathcal S\;\forall T=\poly(\lambda).
\]
\end{definition}

\begin{figure}[!htbp]
\centering
\includegraphics[width=0.92\linewidth]{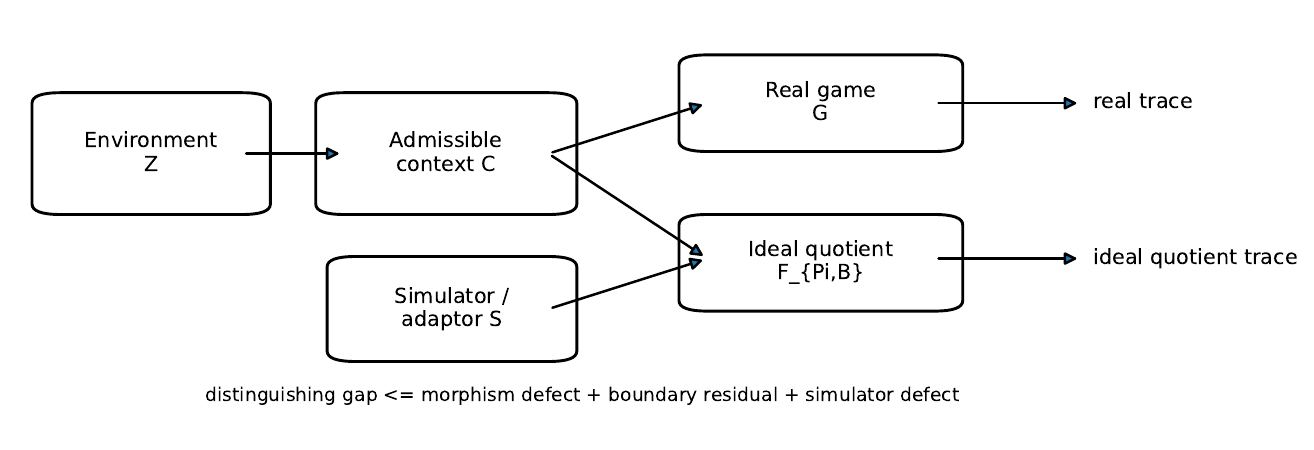}
\caption{Real/ideal quotient security under an environment-controlled context. The simulator/adaptor mediates between a real hidden-state game and an ideal quotient functionality while the context schedules sessions, observers, oracle forwarding, and post-processing.}
\label{fig:realideal}
\end{figure}
\FloatBarrier

\begin{theorem}[Observer-quotient emulation]
Let $G$ be a real interactive OQSG and let $\mathcal F_{\Pi,B}$ be its ideal quotient functionality. Suppose that over horizon $T$ the following predictable bounds hold for every admissible context, environment, and observer scheduler: observer defects $\delta_{\mathrm{obs},t}$, kernel defects $\delta_{K,t}$, post-processing defects $\delta_{\mathrm{post},t}$, interface defects $\delta_{\mathrm{sim},t}$ including abort and leakage-label defects, statistical terms $\eta_t$, and a scalar boundary functional with
\[
\E[Z_{t+1}\mid\cF_t]\le Z_t-a_t\phi(Z_t)+\chi_t,
\qquad
\Adv_t\le \rho_t(\E Z_t)+\eta_t.
\]
Then $G$ emulates $\mathcal F_{\Pi,B}$ with
\[
\epsilon(\lambda,T)
\le
\sum_{t<T}(\delta_{\mathrm{obs},t}+\delta_{K,t}+\delta_{\mathrm{post},t}+\delta_{\mathrm{sim},t}+\eta_t)
+\rho_T(\E Z_T).
\]
If $\phi(z)\ge\mu z$, $a_t\ge a>0$, and $\chi_t\le\bar\chi$, then $\rho_T(\E Z_T)$ is bounded by $\rho_T(\bar\chi/(a\mu))$ up to the transient term. If $\sum_t\chi_t<\infty$ and $\sum_t a_t=\infty$, the boundary contribution vanishes on the limiting execution.
\end{theorem}

\begin{proof}
Construct hybrids $H_0,\ldots,H_{T+2}$. $H_0$ is the real execution. For $1\le j\le T$, $H_j$ replaces the first $j$ scheduled transition/output pairs by their quotient-compatible ideal counterparts, with the same public prefix, same nonuniform advice, same session identifiers, same oracle transcript, and same scheduled observer choices. $H_{T+1}$ replaces the remaining interface behavior: representative sampling, abort synchronization, oracle formatting, and leakage-label consistency. $H_{T+2}$ applies the terminal scalar boundary replacement. The scheduling, representative-sampling, abort, and label lemmas justify the corresponding hops; post-processing closure absorbs the maps $g_t$; the wrapper lemma turns any adjacent separation into one admissible next-hop distinguisher. Summing the indexed hop defects gives the displayed bound.
\end{proof}

\begin{corollary}[Expanded hybrid localization]
If an admissible environment distinguishes $\Real$ from $\Ideal$ with advantage $\epsilon$ exceeding
\[
\Delta_T=\sum_{t<T}(\delta_{\mathrm{obs},t}+\delta_{K,t}+\delta_{\mathrm{post},t}+\delta_{\mathrm{sim},t}+\eta_t)+\rho_T(\E Z_T),
\]
then the averaging lemma and wrapper construction produce an admissible algorithm separating one adjacent hop with advantage at least
\[
\frac{\epsilon-\Delta_T}{T+2}.
\]
The separated hop is one of the following indexed events: observer pullback failure $\delta_{\mathrm{obs},t}$, quotient-kernel failure $\delta_{K,t}$, post-processing closure failure $\delta_{\mathrm{post},t}$, simulator/interface failure $\delta_{\mathrm{sim},t}$, statistical leakage error $\eta_t$, or boundary-to-advantage failure $\rho_T(\E Z_T)$.
\end{corollary}

\begin{proof}
Apply the indexed-hybrid averaging lemma to $H_0,\ldots,H_{T+2}$ after subtracting the declared defects from the telescoping bound. The wrapper construction supplies the admissible adjacent-hop algorithm.
\end{proof}

\section{Primitive-Level Instantiations}

The definitions above become concrete once the observer family is instantiated by a standard security experiment. This section records three such instantiations.

\subsection{IND-CPA encryption with adaptive timing observers}

Let $\mathsf E=(\mathsf{KeyGen},\mathsf{Enc},\mathsf{Dec})$ be a randomized encryption scheme. In the left-or-right experiment, the adversary has advice $a_\lambda$, makes at most $Q=Q(\lambda)$ encryption queries, and adaptively chooses message pairs $(m_{0,i},m_{1,i})$ of equal length. Query $i$ has implementation state
\[
    x_i=(K,r_i,m_{0,i},m_{1,i},C_i,\theta_i,\mathsf{sid}_i),
\]
where $K\leftarrow\mathsf{KeyGen}(1^\lambda)$, $r_i$ is encryption randomness, $C_i=\mathsf{Enc}_K(m_{b,i};r_i)$, $\theta_i$ is the timing state, and $\mathsf{sid}_i$ is the session identifier. The ciphertext observer is
\[
    \Pi_{ct}:\qquad \pi_{ct}(x_i)=(\mathsf{sid}_i,C_i).
\]
The timing-refined observer family contains channels
\[
    \Pi_{ct+time}:\qquad \pi_{time,j}(x_i)=(\mathsf{sid}_i,C_i,\tau_j(K,r_i,m_{b,i},\theta_i)+N_{i,j}),
\]
where $j$ ranges over measurement sites selected by an adaptive scheduler and $N_{i,j}$ is conditionally sub-Gaussian with proxy variance $\sigma_j^2\ge\sigma_{\min}^2>0$. Profiling uses $n_{prof}$ calibration samples per site and gives a mean estimator with error envelope $\Delta_{profile,j}(n_{prof},\delta)$ on the simultaneous calibration event.

For query $i$, define the timing boundary
\[
B_{time,i}=\sup_{j\in J_i}\left|\E[\tau_j(K,r_i,m_{0,i},\theta_i)]-\E[\tau_j(K,r_i,m_{1,i},\theta_i)]\right|.
\]
For Gaussian timing noise with variance $\sigma_j^2$, the total-variation distance between the timing coordinates is
\[
2\Phi\!\left(\frac{|\Delta\mu_j|}{2\sigma_j}\right)-1\le \frac{|\Delta\mu_j|}{2\sigma_j},
\]
and the same linear upper bound follows for sub-Gaussian timing after replacing the Gaussian term by the declared concentration envelope. Thus
\[
    \rho_i(z)=\min\{1,z/(2\sigma_{\min})\}+\Delta_{profile,i}(n_{prof},\delta)
\]
is a concrete boundary-to-advantage map for one scheduled timing coordinate.

\begin{theorem}[IND-CPA with adaptive timing leakage]
Let $A$ be a nonuniform PPT adversary of size $q$ making at most $Q$ encryption queries, and let $S$ be an adaptive timing scheduler that selects sites from the public prefix. Suppose the ciphertext-only scheme has IND-CPA advantage $\Adv^{ind\text{-}cpa}_{\mathsf E}(q',Q')$ against reductions of size $q'$ and query count $Q'$. Suppose further that the timing boundaries have adapted surrogates $Z_i$ satisfying
\[
\E[Z_{i+1}\mid\cF_i]\le Z_i-a_i\phi(Z_i)+\chi_i,
\qquad
\Adv_i^{time}\le \rho_i(\E Z_i)+\eta_i.
\]
Then the ciphertext-plus-timing experiment satisfies
\[
\Adv_{\mathsf E}^{ind\text{-}cpa+time}(A,S)
\le
Q\,\Adv^{ind\text{-}cpa}_{\mathsf E}(q+\poly(Q,\lambda),Q)
+\sum_{i=1}^Q\bigl(\rho_i(\E Z_i)+\eta_i\bigr)
+\Delta_{profile}(n_{prof},\delta),
\]
where $\Delta_{profile}=\sum_{i\le Q}\sup_{j\in J_i}\Delta_{profile,j}(n_{prof},\delta)$ on the calibration event. If $\phi(z)\ge\mu z$, $a_i\ge a>0$, and $\chi_i\le\bar\chi$, the timing contribution is bounded by
\[
Q\,\rho\!\left(\frac{\bar\chi}{a\mu}\right)+\sum_{i=1}^Q\eta_i+\Delta_{profile}(n_{prof},\delta)
\]
up to the transient term. If $\sum_i\chi_i<\infty$ and $\sum_i a_i=\infty$, the boundary contribution vanishes along the hardened execution.
\end{theorem}

\begin{proof}
Use hybrids $G_0,\ldots,G_{Q+2}$. $G_0$ is the real left-or-right experiment with timing. For $1\le h\le Q$, $G_h$ replaces the first $h$ ciphertexts by the ordinary IND-CPA hybrid ciphertexts while preserving session identifiers, scheduler state, timing labels, and oracle transcript. The reduction for hop $h$ samples $h$, embeds its left-or-right challenge at the $h$th encryption query, answers the remaining encryption queries honestly, forwards $A$'s public transcript to the scheduler, and simulates timing labels from the profiled timing law. Its advice is $A$'s advice together with the sampled hop index, and its size and query counts are $q+\poly(Q,\lambda)$ and at most $Q$ encryption-oracle simulations. The factor $Q$ follows from hybrid averaging. $G_{Q+1}$ replaces real timing samples by the ideal timing samples conditioned on the quotient class; the TV conversion and profiling envelope give $\sum_i(\rho_i(\E Z_i)+\eta_i)+\Delta_{profile}$. $G_{Q+2}$ applies the residual-floor replacement from the scalar recurrence. Summing the hops gives the bound.
\end{proof}

\subsection{Deterministic encryption with entropy leakage}

For deterministic encryption, security depends on the message source. Let $M$ be sampled with side information $A$ and let $H_\infty^\epsilon(M\mid A)\ge h$. Leakage is accounted for in two compatible ledgers. In the alphabet-size ledger, a leakage variable $L$ with $|\operatorname{range}(L)|\le2^\ell$ spends at most $\ell$ bits. In the smooth max-information ledger, leakage with $I_\infty^\epsilon(M;L\mid A)\le\ell$ spends at most $\ell$ bits under the same smoothing convention. We use the convention that a step with smoothing parameters $(\epsilon,\epsilon')$ satisfies
\[
H_\infty^{\epsilon+\epsilon'}(M\mid A,L)\ge h-\ell-\log(1/\epsilon').
\]
Let $h_{req}$ be the entropy threshold required by the deterministic-encryption theorem used for the base scheme, and define the entropy boundary
\[
    B_H=(h-\ell)-h_{req}.
\]

\begin{center}
\scriptsize
\setlength{\tabcolsep}{3pt}
\begin{tabular}{@{}p{0.16\linewidth}p{0.17\linewidth}p{0.20\linewidth}p{0.15\linewidth}p{0.15\linewidth}@{}}
\toprule
Operation & Ledger update & Entropy after step & Smoothing & Security condition\\
\midrule
Alphabet leakage & $|L|\le2^\ell$ & $h\leftarrow h-\ell$ & $\epsilon$ fixed & $h\ge h_{req}$\\
Max-info leakage & $I_\infty^\epsilon\le\ell$ & $h\leftarrow h-\ell$ & $\epsilon\leftarrow\epsilon+\epsilon'$ & $h\ge h_{req}$\\
Refresh & add $r$ modeled bits & $h\leftarrow h+r$ & source model & $h\ge h_{req}$\\
One step & leak $\ell_t$, refresh $r_t$ & $h_{t+1}=h_t+r_t-\ell_t$ & add $\epsilon_t'$ & $B_H(t)\ge0$\\
\bottomrule
\end{tabular}
\end{center}
Here $r_t^{-}=0$ unless the source model also spends entropy through conditioning or rejection; in the common renewal model the update is $h_{t+1}=h_t+r_t-\ell_t$.

\begin{proposition}[Entropy-spending boundary functional]
Let $Z_t=(h_{req}-H_t)_+$ be the entropy-deficit functional. If step $t$ leaks at most $\ell_t$ bits in either ledger and refreshes $r_t$ accounted bits, then
\[
Z_{t+1}\le Z_t-(r_t-\ell_t)+\xi_t,
\]
where $\xi_t$ records smoothing and estimation slack. When $\sum_t(r_t-\ell_t)=\infty$ and $\sum_t\xi_t<\infty$, the distributional observer quotient approaches the secure entropy region. When leakage exceeds refresh persistently, the residual floor is the remaining entropy deficit and the deterministic-encryption theorem converts it into a distinguishing bound.
\end{proposition}

\subsection{Nonce-refreshing authenticated-encryption state machines}

An authenticated-encryption implementation may have state
\[
    x=(K,N,c,\theta,\mathsf{sid}),
\]
where $K$ is the key, $N$ the nonce, $c$ a counter, $\theta$ a microarchitectural state, and $\mathsf{sid}$ the session identifier. A transcript observer sees $(\mathsf{sid},N,C,T,\mathsf{abort})$; timing, cache, power, and electromagnetic observers additionally see scheduled leakage samples. A nonce-refresh transition is a hidden continuation when it changes $(N,c,\theta)$ while preserving the transcript quotient. Compatibility requires
\[
(K,N,c,\theta,\mathsf{sid})\equiv_\Pi(K,N',c',\theta',\mathsf{sid}')
\Longrightarrow
q_{\Pi\#}K(x,\cdot)=q_{\Pi\#}K(x',\cdot)
\]
up to the declared kernel defect. The scalar boundary term tracks nonce-collision probability, leakage-state drift, and abort-label discrepancy. The real/ideal theorem then gives a single advantage bound whose terms correspond to nonce misuse, observer incompatibility, side-channel residual, and simulator error.

\section{Quotient Dynamics and Trace Lifting}

\subsection{Compatibility with the quotient}

A transition kernel respects an observer quotient when it does not split equivalent states into distinguishable future distributions.

\begin{definition}[Quotient-compatible kernel]
Let $K:X\to\DeltaS(X)$ be a Markov kernel. It is $\Pi$-compatible when for all $x\eqobs x'$ and every measurable union $B$ of equivalence classes in $X/\Pi$,
\[
    K(x)(B)=K(x')(B).
\]
Equivalently, the pushforward kernels $q_{\Pi\#}K(x)$ and $q_{\Pi\#}K(x')$ on $X/\Pi$ are equal whenever $x\eqobs x'$.
\end{definition}

\begin{proposition}[Induced quotient dynamics]
If $K$ is $\Pi$-compatible, then there exists a unique kernel $\bar K:X/\Pi\to\DeltaS(X/\Pi)$ such that
\[
    \bar K(q_{\Pi}(x)) = q_{\Pi\#}K(x)
\]
for every $x\in X$.
\end{proposition}

\begin{proof}
Define $\bar K([x])=q_{\Pi\#}K(x)$. Compatibility says this definition is independent of the representative $x$ of the class $[x]$. Uniqueness follows from surjectivity of $q_{\Pi}$.
\end{proof}

\subsection{Trace lifting}

The next theorem is the basic sufficient condition that allows local quotient reasoning to become trace reasoning.

\begin{theorem}[Quotient trace lifting]\label{thm:trace-lifting}
Let $K_0,K_1$ be $\Pi$-compatible kernels on $X$ with induced quotient kernels $\bar K_0,\bar K_1$. Suppose $q_{\Pi}(x_0^0)=q_{\Pi}(x_0^1)$ and $\bar K_0=\bar K_1$. Then for every horizon $T$ and every non-adaptive observer sequence $\pi_0,\ldots,\pi_T\in\Pi$, the transcript distributions under $\mathsf{Exp}_0$ and $\mathsf{Exp}_1$ are identical. Consequently, the unbounded distinguishing advantage is zero.
\end{theorem}

\begin{proof}
Because $\bar K_0=\bar K_1$ and the initial quotient states are equal, the quotient processes $q_{\Pi}(X_t^0)$ and $q_{\Pi}(X_t^1)$ have the same law for every finite horizon. If two raw states lie in the same quotient class, every observer in $\Pi$ produces the same output distribution by definition. Therefore every cylinder probability of the transcript process agrees between the two experiments. Equality of all finite cylinder probabilities gives equality of transcript distributions. The total variation distance between the two transcript distributions is zero, hence every test has zero advantage.
\end{proof}

\begin{corollary}[Adaptive observer lifting under policy compatibility]
Assume the hypotheses of Theorem~\ref{thm:trace-lifting} except that the observer at time $t$ is chosen by an adaptive policy $\rho_t$ that depends only on the previous transcript and selects $\pi_t\in\Pi$. Then the conclusion remains true.
\end{corollary}

\begin{proof}
Condition on any transcript prefix with positive probability. By the non-adaptive lifting argument, the two experiments assign the same probability to that prefix. The policy selects the same next observer distribution as a function of that prefix. The next output law depends only on the quotient state law, which is equal in both experiments. Induction over $t$ proves equality of all adaptive transcript distributions.
\end{proof}

\subsection{Approximate lifting}

Exact equality is too strong for computational security and noisy sensors. The approximate version replaces equality by an error accumulation bound.

\begin{definition}[One-step quotient defect]
For kernels $K_0,K_1$ and observer class $\Pi$, define the quotient defect at tolerance zero by
\[
    \delta_{\Pi}(K_0,K_1)=\sup_{x\eqobs x'}\TV\big(q_{\Pi\#}K_0(x),q_{\Pi\#}K_1(x')\big).
\]
\end{definition}

\begin{theorem}[Telescoping trace bound]
Assume that at each time step the quotient transition defect is at most $\delta$ and the initial quotient distributions are within total variation distance $\eta$. Then, for any fixed horizon $T$ and any transcript observer sequence from $\Pi$,
\[
    \TV(\mathsf{Law}_0(Z_{0:T}),\mathsf{Law}_1(Z_{0:T}))\le \eta + T\delta.
\]
\end{theorem}

\begin{proof}
Couple the two quotient processes optimally at time $0$, failing with probability at most $\eta$. At each step, condition on the coupled quotient states being equal. The transition defect allows a maximal coupling of the next quotient states that fails with probability at most $\delta$. By the union bound, the probability that the quotient paths decouple by time $T$ is at most $\eta+T\delta$. Conditional on no decoupling, all observer outputs have the same conditional law. Therefore the transcript total variation distance is bounded by the decoupling probability.
\end{proof}

This theorem is intentionally elementary. It gives a conservative but transparent bound. Sharper bounds can be obtained by contraction coefficients, mixing, or martingale arguments, but the linear form already exposes the relevant security warning: small local leakage can become non-small over repeated interaction.

\section{Hidden Continuations}

\subsection{Separation between raw and observed fixedness}

\begin{theorem}[Non-injective observers admit hidden continuations]
Let $X$ be a set and let $\pi:X\to Y$ be a deterministic observer that is not injective. Then there exists a deterministic transition $T:X\to X$ and a state $x\in X$ such that $T(x)\ne x$ but $\pi(T(x))=\pi(x)$. Thus raw fixedness and observational fixedness are distinct.
\end{theorem}

\begin{proof}
Because $\pi$ is not injective, choose $x\ne x'$ with $\pi(x)=\pi(x')$. Define $T(x)=x'$ and define $T$ arbitrarily on all other states, for example $T(u)=u$. Then $T(x)\ne x$ while the observer output is unchanged.
\end{proof}

\begin{corollary}[Observer kernels have invisible directions]
Let $C\in\R^{p\times n}$ be a linear output map with nontrivial kernel. For every $v\in\ker C$ and every state $x\in\R^n$, the transition $T(x)=x+v$ is hidden from the observer $y=Cx$.
\end{corollary}

\begin{proof}
$C(x+v)=Cx+Cv=Cx$.
\end{proof}

\subsection{One-step fixedness is not trace security}

Observational fixedness at one time does not imply security over time. The following finite counterexample is finite and deterministic.

\begin{theorem}[One-step observational fixedness does not imply trace security]
There exist a finite state space $X$, observer $\pi$, transition $T$, and two initial states $x_0,x_1$ such that $\pi(x_0)=\pi(x_1)$ and $\pi(T(x_0))=\pi(T(x_1))$, but the length-three transcripts distinguish the two worlds with advantage one.
\end{theorem}

\begin{proof}
Let $X=\{a_0,a_1,b_0,b_1,c_0,c_1\}$. Let the observer output $0$ on $a_0,a_1,b_0,b_1$ and output $i$ on $c_i$ for $i\in\{0,1\}$. Define $T(a_i)=b_i$, $T(b_i)=c_i$, and $T(c_i)=c_i$. Let the two initial states be $x_i=a_i$. At times $0$ and $1$, both worlds produce observer output $0$. At time $2$, world $i$ produces output $i$. Therefore a distinguisher that reads the third output identifies the world with probability one. The one-step equality of observations did not control the future continuation.
\end{proof}

\begin{remark}
This theorem is the formal version of a common side-channel failure. A test may show that one operation leaks no immediate signal, while a later operation reveals a state variable that the earlier operation changed silently.
\end{remark}

\subsection{Hidden continuation monoids}

For deterministic dynamics, hidden continuations compose when they remain inside the same observer class.

\begin{definition}[Hidden continuation set]
For a deterministic transition semigroup $\cT$ acting on $X$, define
\[
    \cH_{\Pi}(x)=\{S\in\cT:Sx\eqobs x\}.
\]
The strict hidden continuation set is $\cH_{\Pi}^{\ne}(x)=\{S\in\cH_{\Pi}(x):Sx\ne x\}$.
\end{definition}

\begin{proposition}[Local monoid property]
If $S,R\in\cT$ satisfy $Sx\eqobs x$ and $RSx\eqobs Sx$, then $RS\in\cH_{\Pi}(x)$. In particular, if every transition in a subsemigroup preserves each observer class, then the hidden continuations at $x$ form a submonoid.
\end{proposition}

\begin{proof}
The first statement follows by transitivity of $\eqobs$: $RSx\eqobs Sx\eqobs x$. The submonoid statement follows because the identity is observationally fixed and composition preserves the equivalence class by assumption.
\end{proof}

This algebraic viewpoint is useful because it distinguishes accidental one-step invisibility from a stable invisible mode. A stable invisible mode is dangerous for security monitoring: it can accumulate arbitrary internal change without leaving the quotient.

\section{Composability and Post-Processing}

\subsection{Post-processing closure}

Security definitions not be invalidated by transforming the adversary's view through a weaker channel. This is the data-processing principle in observer form.

\begin{theorem}[Observer post-processing]
Let $\pi:X\to\DeltaS(Y)$ be an observer and let $R:Y\to\DeltaS(W)$ be any Markov kernel. Define $R\circ\pi:X\to\DeltaS(W)$ by
\[
    (R\circ\pi)(w\mid x)=\sum_{y\in Y}R(w\mid y)\pi(y\mid x).
\]
Then
\[
    \TV((R\circ\pi)(\cdot\mid x),(R\circ\pi)(\cdot\mid x'))\le \TV(\pi(\cdot\mid x),\pi(\cdot\mid x')).
\]
Consequently, adding only post-processings of existing observers cannot refine the observer quotient.
\end{theorem}

\begin{proof}
For finite spaces,
\[
\begin{aligned}
2\TV(R\pi_x,R\pi_{x'})
&=\sum_w\left|\sum_y R(w\mid y)(\pi(y\mid x)-\pi(y\mid x'))\right|\\
&\le \sum_w\sum_y R(w\mid y)|\pi(y\mid x)-\pi(y\mid x')|\\
&=\sum_y |\pi(y\mid x)-\pi(y\mid x')|.
\end{aligned}
\]
Divide by two.
\end{proof}

\subsection{Sequential composition}

If independent sessions each leak a small amount, repeated interaction can amplify leakage. The quotient language gives the usual hybrid argument.

\begin{theorem}[Hybrid composition bound]
Suppose a protocol is executed $q$ times and session $j$ changes the transcript distribution by at most $\eps_j$ in total variation under observer class $\Pi$, even conditioned on all previous transcript prefixes. Then the total distinguishing advantage over the $q$-session transcript is at most $\sum_{j=1}^{q}\eps_j$.
\end{theorem}

\begin{proof}
Construct hybrids $H_0,\ldots,H_q$, where $H_j$ uses world $1$ for the first $j$ sessions and world $0$ for the remaining sessions. By assumption, adjacent hybrids differ by at most $\eps_j$ in total variation. The triangle inequality gives
\[
    \TV(H_0,H_q)\le\sum_{j=1}^{q}\TV(H_{j-1},H_j)\le\sum_{j=1}^{q}\eps_j.
\]
The distinguishing advantage of any test is bounded by total variation.
\end{proof}

\begin{remark}
The theorem identifies the quantities controlling composition: the conditional leakage of each continuation given the transcript already produced.
\end{remark}

\subsection{Observer refinement monotonicity}

\begin{proposition}[Refinement monotonicity]
If $\Pi\subseteq\Pi'$, then $d_{\Pi}(x,x')\le d_{\Pi'}(x,x')$ for all $x,x'$. Hence $x\equiv_{\Pi'}x'$ implies $x\equiv_{\Pi}x'$, and $X/\Pi'$ is a refinement of $X/\Pi$.
\end{proposition}

\begin{proof}
The supremum over the smaller class cannot exceed the supremum over the larger class. The equivalence and quotient statements follow immediately.
\end{proof}

This elementary result is the mathematical reason that adding a side channel can invalidate a proof. A theorem proved under $\Pi$ does not automatically survive under $\Pi'$.

\section{Operational Side-Channel Observer Classes}

Observer refinement is the security operation behind side-channel analysis. A transcript observer induces a coarse quotient; timing, cache, power, electromagnetic, and profiled observers refine it by adding measurement channels. The refinement is measurable once the leakage variables, calibration protocol, scheduler regime, and confidence envelope are fixed.

\begin{definition}[Side-channel observer families]
For an implementation state $x$ and public transcript $T(x)$, define
\begin{align*}
\Pi_{tr} &= \{x\mapsto T(x)\},\\
\Pi_{time} &= \{x\mapsto (T(x),\tau_j(x)+N_j): j\in J_t\},\\
\Pi_{cache} &= \{x\mapsto (T(x),\mathsf{set}_j(x)\oplus N_j): j\in J_c\},\\
\Pi_{power} &= \{x\mapsto (T(x),\langle w_j,p(x)\rangle+N_j): j\in J_p\},\\
\Pi_{em} &= \{x\mapsto (T(x),\langle u_j,e(x)\rangle+N_j): j\in J_e\},\\
\Pi_{prof} &= \{x\mapsto (T(x),\psi_\theta(\ell(x))+N_\theta): \theta\in\Theta\}.
\end{align*}
The functions $\tau,\mathsf{set},p,e$ denote timing, cache-set, power, and electromagnetic features, while $\psi_\theta$ is a profiled statistic trained on calibration traces. If $\Pi\subseteq\Pi'$, then $d_\Pi\le d_{\Pi'}$ and the quotient under $\Pi'$ refines the quotient under $\Pi$.
\end{definition}

\begin{definition}[Profiled boundary calibration]
A calibration instance consists of transcript-equivalent pairs $(x,x')$, a matched continuation schedule $\Gamma$, calibration traces independent of test traces, a feature map $f_\theta$, and a statistic
\[
B_{prof}(x,x')=d_{stat}\bigl(f_\theta(\ell_\Gamma(x)),f_\theta(\ell_\Gamma(x'))\bigr).
\]
The statistic may be total variation, Mahalanobis distance, classifier advantage, or a likelihood-ratio score. In the scalar sub-Gaussian mean-gap model with proxy variance $\sigma^2$ and model mismatch $\Delta_{model}$,
\[
\eta_{prof}(n,\delta)=\sqrt{2\sigma^2\log(2/\delta)/n}+\Delta_{model}
\]
upper-bounds the estimation error for one pair. A family of $m$ tested pairs uses the simultaneous envelope
\[
\eta_{m}(n,\delta)=\sqrt{2\sigma^2\log(2m/\delta)/n}+\Delta_{model}.
\]
\end{definition}

\begin{definition}[Leakage regimes]
Leakage is stationary when the conditional law of the feature depends only on the representative and scheduled observer. It is drifting when an unobserved parameter $\nu_t$ changes with bounded variation $\sum_t\|\nu_{t+1}-\nu_t\|\le D$. It is scheduler-induced when the environment's session order, cache warmup, nonce schedule, throttling, or oracle pattern changes the leakage law. The validation envelope is $\eta_m$ in the stationary regime, $\eta_m+D$ in the drifting regime, and the supremum over the admissible schedule family in the scheduler-induced regime.
\end{definition}

\begin{figure}[!htbp]
\centering
\includegraphics[width=0.86\linewidth]{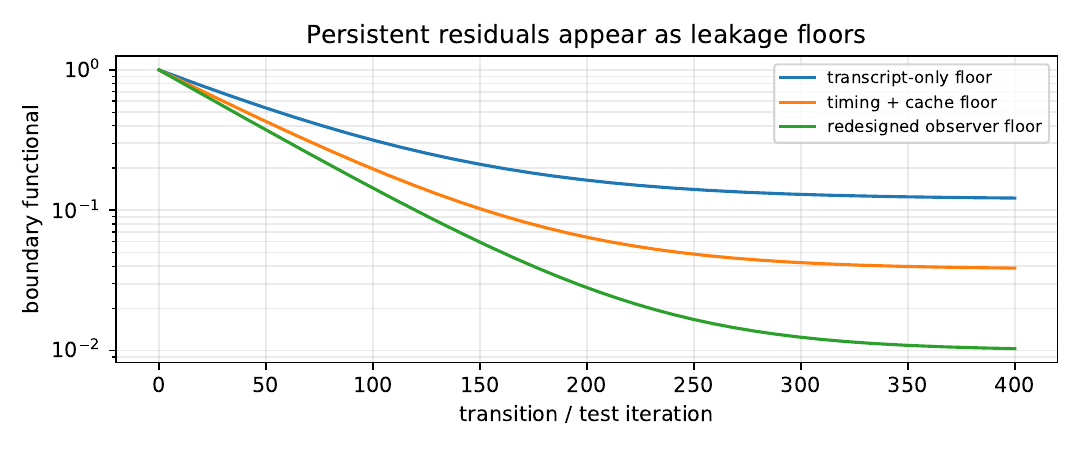}
\caption{Leakage floors under observer refinement. Persistent statistical residuals keep the scalar functional above zero, while implementation redesign or observer redesign lowers the resulting advantage bound.}
\label{fig:side-floor}
\end{figure}
\FloatBarrier

\begin{theorem}[Finite-state leakage refinement]
Let $\mathcal P_{tr}$ be the transcript quotient partition of a finite implementation state machine. Let $\mathcal N$ be a set of representative pairs inside blocks of $\mathcal P_{tr}$. For each pair, run $n$ matched continuations under the declared schedule family and compute $\widehat B_{prof}(x,x')$. Under the calibration assumptions above, with probability at least $1-\delta$,
\[
|\widehat B_{prof}(x,x')-B_{prof}(x,x')|\le \eta_m(n,\delta)
\]
for all $(x,x')\in\mathcal N$. The refinement rule
\[
\widehat B_{prof}(x,x')-\eta_m(n,\delta)>r^\star
\]
splits exactly those tested pairs whose lower confidence bound exceeds the residual floor. Every unsplit tested pair has profiled pseudometric at most $r^\star+2\eta_m(n,\delta)$ on the same event.
\end{theorem}

\begin{proof}
Apply the sub-Gaussian concentration inequality to each pair and union-bound over $\mathcal N$. The lower-confidence split rule gives separation above $r^\star$ for split pairs. If a pair is not split, its empirical score is at most $r^\star+\eta_m$, and the simultaneous event gives a true score at most $r^\star+2\eta_m$.
\end{proof}

\begin{proposition}[Regression criterion for implementation revisions]
Let $\mathcal R_k$ be an implementation revision and let $\widehat B_k$ be the profiled boundary matrix over the same state-pair set, schedule family, and observer family. Revision $k+1$ preserves the tested observer quotient at floor $r^\star$ when
\[
\max_{(x,x')\in\mathcal N}\bigl(\widehat B_{k+1}(x,x')-\eta_m(n,\delta)\bigr)_+
\le r^\star
\]
for all pairs that were unsplit in revision $k$, and when every newly split pair is recorded as a refinement of the quotient partition before the compatibility and emulation bounds are reused.
\end{proposition}

\begin{proof}
The first condition preserves the unsplit relation on the simultaneous concentration event. The second condition prevents reuse of a proof over a coarser partition after the observer class has made a distinction measurable.
\end{proof}

\subsection{Validation pipeline}

A finite-state validation run is specified by the following data.
\begin{enumerate}[leftmargin=2em]
\item A state-pair generator enumerates or samples $(x,x')$ with $T(x)=T(x')$ under $\Pi_{tr}$.
\item A matched continuation schedule $\Gamma$ fixes session identifiers, nonce refreshes, cache warmup, adversarial query order, and post-processing.
\item Timing, cache, power, electromagnetic, or profiled feature traces are acquired under the chosen leakage regime.
\item A statistic $\widehat B_{prof}$ and simultaneous envelope $\eta_m(n,\delta)$ are computed.
\item Blocks of the transcript quotient are split according to the lower-confidence rule.
\item Quotient compatibility and the scalar boundary recurrence are recomputed on the refined partition.
\end{enumerate}
In pseudocode, the finite-state test is
\[
\begin{array}{l}
\textsf{Pairs}\leftarrow\{(x,x'):T(x)=T(x')\};\quad
\textsf{Traces}\leftarrow\textsf{RunMatched}(\textsf{Pairs},\Gamma,n);\\
\widehat B\leftarrow\textsf{ProfiledStatistic}(\textsf{Traces});\quad
\eta\leftarrow\textsf{Envelope}(n,|\textsf{Pairs}|,\delta);\\
\textsf{Split}\leftarrow\{(x,x'):\widehat B(x,x')-\eta>r^\star\};\quad
\mathcal P'\leftarrow\textsf{Refine}(\mathcal P_{tr},\textsf{Split}).
\end{array}
\]

\begin{example}[Finite-state side-channel case study]
Consider two implementation states $s_1,s_2$ that emit the same transcript symbol and later return to the same public control state. A transcript observer has $s_1\equiv_{\Pi_{tr}}s_2$. A profiled power observer returns
\[
L_j=\alpha\operatorname{HW}(S_j)+N_j,
\qquad N_j\sim\mathcal N(0,\sigma^2).
\]
With $n$ matched traces per representative, the simultaneous lower bound for the mean displacement is
\[
\underline\Delta=\left(|\widehat\mu_1-\widehat\mu_2|-2\sigma\sqrt{2\log(4m/\delta)/n}\right)_+.
\]
The profiled lower functional is $\underline B_{prof}=\underline\Delta^2/(2\sigma^2)$. If $\underline B_{prof}>r^\star$, the quotient block splits under $\Pi_{tr}\cup\Pi_{prof}$. If $\underline B_{prof}\le r^\star$, the tested continuation remains inside the declared leakage floor for that observer and schedule family.
\end{example}

\section{Control-Theoretic Interpretation}

\subsection{Linear observability as equality reflection}

Consider a discrete-time linear system
\[
    x_{t+1}=Ax_t,\qquad y_t=Cx_t,
\]
with $x_t\in\R^n$ and $y_t\in\R^p$. The finite-horizon observation map is
\[
    \cO_T x = \begin{bmatrix} Cx \\ CAx \\ \vdots \\ CA^{T}x\end{bmatrix}.
\]
Two initial states are observer-equivalent over horizon $T$ when $\cO_Tx=\cO_Tx'$.

\begin{definition}[Equality-reflecting observer]
An observer map $\cO:X\to Y$ is equality-reflecting when $\cO(x)=\cO(x')$ implies $x=x'$.
\end{definition}

\begin{theorem}[Observability equals equality reflection]
For the linear system above, the horizon-$T$ observer $\cO_T$ is equality-reflecting if and only if
\[
    \operatorname{rank}\begin{bmatrix} C \\ CA \\ \vdots \\ CA^{T}\end{bmatrix}=n.
\]
The hidden continuation directions are exactly $\ker \cO_T$.
\end{theorem}

\begin{proof}
The map $\cO_T$ is linear. It is equality-reflecting if and only if $\cO_Tx=\cO_Tx'$ implies $x-x'=0$, which is equivalent to $\ker\cO_T=\{0\}$. For a linear map from $\R^n$, this is equivalent to rank $n$. If $v\in\ker\cO_T$, then $x$ and $x+v$ produce the same output trajectory over the horizon, so $v$ is a hidden continuation direction. Conversely, every hidden direction satisfies $\cO_Tv=0$.
\end{proof}

\subsection{Stealthy attacks}

Let the attacked system be
\[
    x_{t+1}=Ax_t+Bu_t+Ea_t,\qquad y_t=Cx_t,
\]
where $a_t$ is an adversarial input. A stealthy attack is an input sequence that changes the internal state or cost while preserving the observation transcript.

\begin{proposition}[Finite-horizon stealth condition]
Assume the nominal and attacked systems start from the same state. The attack sequence $a_{0:T-1}$ is output-stealthy over horizon $T$ if and only if
\[
    C\sum_{k=0}^{t-1}A^{t-1-k}Ea_k=0\qquad \text{for all }t=1,\ldots,T.
\]
\end{proposition}

\begin{proof}
Subtract the nominal trajectory from the attacked trajectory. The deviation $e_t$ satisfies $e_0=0$ and $e_{t+1}=Ae_t+Ea_t$, hence
\[
    e_t=\sum_{k=0}^{t-1}A^{t-1-k}Ea_k.
\]
The output deviation is $Ce_t$. Stealth is exactly the condition $Ce_t=0$ for all $t$.
\end{proof}

This proposition is the control version of a side channel. The attacker moves through the kernel of the observer. A security proof that ignores the relevant sensor map cannot detect the attack.

\section{A Worked LTI Sensor-Design Template}
\label{sec:lti-example}

Let
\[
A=\begin{bmatrix}1&1&0\\0&1&0\\0&0&0.8\end{bmatrix},\qquad
C_0=\begin{bmatrix}1&0&0\end{bmatrix},
\]
and take horizon $T=3$. The observability matrix is
\[
\mathcal O_T(C_0)=\begin{bmatrix}C_0\\C_0A\\C_0A^2\end{bmatrix}
=\begin{bmatrix}1&0&0\\1&1&0\\1&2&0\end{bmatrix},
\]
so $\rank\mathcal O_T(C_0)=2$ and $\Ker\mathcal O_T(C_0)=\operatorname{span}\{e_3\}$. The third coordinate is a hidden continuation direction. Adding the row $c_3=[0\;0\;1]$ gives
\[
C_1=\begin{bmatrix}1&0&0\\0&0&1\end{bmatrix},
\qquad \Ker\mathcal O_T(C_1)=\{0\}.
\]
The quotient interpretation is direct: $C_0$ identifies all states that differ by $e_3$, while $C_1$ separates that direction.

\begin{figure}[!htbp]
\centering
\includegraphics[width=0.82\linewidth]{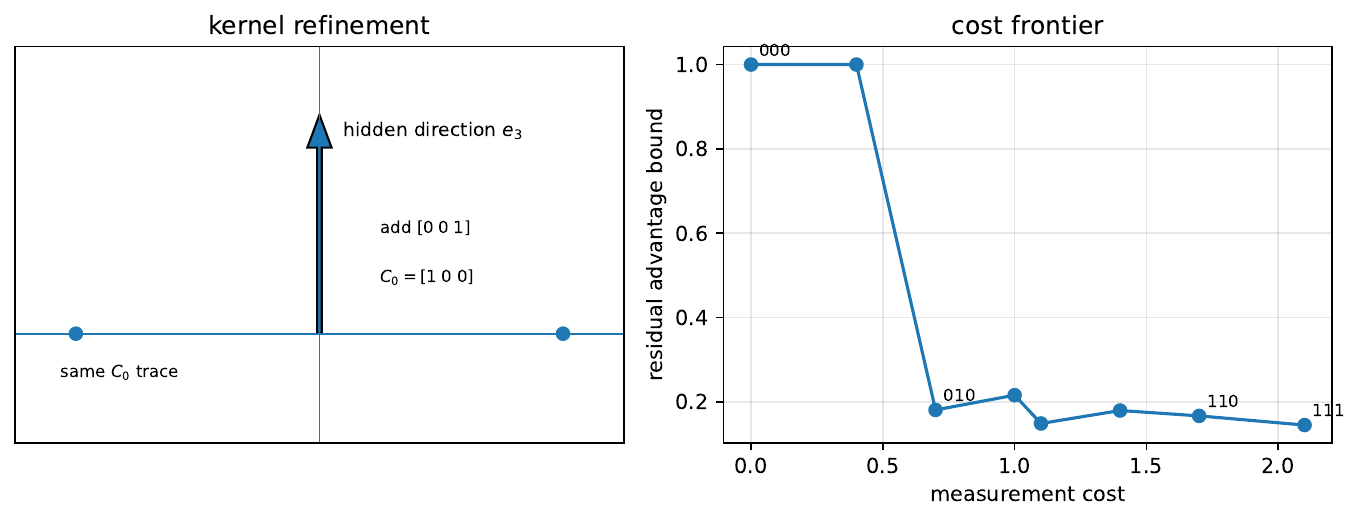}
\caption{LTI observer redesign and cost frontier. The added row removes the hidden-continuation kernel; the relaxed sensor-design audit records the corresponding reduction in residual advantage bound as measurement cost increases.}
\label{fig:lti-design}
\end{figure}
\FloatBarrier

For Gaussian output observations $Y_{0:T}=\mathcal O_T(C)x+N$, $N\sim\mathcal N(0,\Sigma_N)$, define
\[
    B_C(x,x')=\|\Sigma_N^{-1/2}\mathcal O_T(C)(x-x')\|.
\]
The total variation between the two Gaussian output laws is bounded by a monotone function of $B_C$. Hence adding a sensor row changes a control-theoretic quantity, the observability kernel, and a security quantity, the admissible observer pseudometric, in the same step.

\begin{proposition}[Sensor row improvement]
If a candidate row $c$ satisfies $cA^t v\ne0$ for some $v\in\Ker\mathcal O_T(C)$ and some $0\le t\le T$, then adding $c$ strictly reduces $\Ker\mathcal O_T(C)$. Equivalently, the observer partition strictly refines on the span of the tested continuation direction.
\end{proposition}

\begin{proof}
After adding $c$, a hidden vector lies in the old kernel and in the kernel of each new row $cA^t$. The displayed condition excludes at least one old hidden direction, so the intersection kernel is a strict subspace.
\end{proof}

Let $\bar C_1,\ldots,\bar C_m$ be candidate sensor rows and let $s\in[0,1]^m$ be a relaxed sensor-selection vector. Define the affine information matrix
\[
G(s)=\theta I+\sum_{i=1}^m s_iG_i,
\qquad
G_i=\mathcal O_T(\bar C_i)^\top\mathcal O_T(\bar C_i)\succeq0,
\]
with $\theta>0$. A reusable observer-design relaxation is
\[
\begin{aligned}
\min_{s\in[0,1]^m}
&\quad -\alpha\log\det G(s)+\beta\operatorname{tr}\Sigma_e(s)+\gamma c^\top s\\
\text{subject to}
&\quad c^\top s\le C_{max}.
\end{aligned}
\]
The map $s\mapsto\log\det G(s)$ is concave on the positive-definite cone because $G(s)$ is affine and positive definite. Therefore $-\log\det G(s)$ is convex. When $\operatorname{tr}\Sigma_e(s)$ is convex in the same information parametrization and the cost is linear, the relaxation is a convex program. For the original binary kernel-dimension problem, the log-det term is a convex surrogate for removing poorly observed directions, not an exact kernel-dimension objective.

\begin{theorem}[KKT stationarity for relaxed observer design]
Assume the relaxation is convex and satisfies Slater's condition. At an optimum $s^\star$, there are multipliers $\lambda\ge0$, $\nu_i^+\ge0$, and $\nu_i^-\ge0$ such that
\[
-\alpha\Tr(G(s^\star)^{-1}G_i)
+\beta\frac{\partial}{\partial s_i}\Tr\Sigma_e(s^\star)
+\gamma c_i+\lambda c_i+\nu_i^+-\nu_i^-=0,
\]
with $\lambda(c^\top s^\star-C_{max})=0$, $\nu_i^+s_i^\star=0$, and $\nu_i^-(1-s_i^\star)=0$. The stationarity equation compares marginal quotient refinement, estimation improvement, and measurement cost.
\end{theorem}

\begin{proof}
The derivative of $-\log\det G(s)$ in coordinate $s_i$ is $-\Tr(G(s)^{-1}G_i)$. The remaining terms are differentiable by assumption, and the displayed equations are the KKT conditions for the box and budget constraints.
\end{proof}

\begin{theorem}[Observer redesign to bounded advantage reduction]
Let $s,s'\in[0,1]^m$ satisfy $G(s')\succeq G(s)$ and $\bar\chi(s')\le\bar\chi(s)$. For a tested continuation set $V$ and statistical tolerance $\epsilon_{stat}>0$, define
\[
R(s)=\frac{\bar\chi(s)}{a\mu}
+\sup_{v\in V}\frac{\epsilon_{stat}}{\max\{\|G(s)^{1/2}v\|,\theta^{1/2}\|v\|\}}.
\]
If $\rho$ is nondecreasing, the worst-case distinguishing advantage for continuations in $V$ is bounded by $\rho(R(s))$. If $R(s')<R(s)$, then the refined design satisfies
\[
    \sup_{v\in V}\Adv_{s'}(v)\le \rho(R(s'))<\rho(R(s))
\]
whenever $\rho$ is strictly increasing on the interval between the two residuals.
\end{theorem}

\begin{proof}
The Loewner order gives $\|G(s')^{1/2}v\|\ge\|G(s)^{1/2}v\|$ for all $v$, so the statistical uncertainty term cannot increase. The scalar boundary recurrence contributes $\bar\chi(s)/(a\mu)$, which is nonincreasing by assumption. The map $\rho$ converts the resulting residual comparison into an advantage comparison.
\end{proof}

\section{Utility--Leakage Observer Design}

\subsection{Observer design as an optimization problem}

A defender often chooses what to reveal. Revealing more can improve utility, monitoring, or controllability; revealing less can reduce leakage. Let $S$ be a secret, $X$ an internal state correlated with $S$, and $Y$ an observer output generated by a channel $W:Y\mid X$. Let $U(W)$ be a utility functional, $L(W)$ a leakage functional, and $C(W)$ a measurement or computation cost. The generic observer-design problem is
\[
    \max_{W\in\cW}\; U(W)-\lambda L(W)-\mu C(W),\qquad \lambda,\mu\ge 0.
\]
The design variable is therefore an observer channel, and the same channel controls both utility and leakage.

\begin{theorem}[Existence of an optimal finite observer]
Let $X,Y$ be finite and let $\cW$ be a nonempty compact subset of the channel simplex $\{W:Y\mid X\}$. If $U,L,C$ are continuous on $\cW$, then the observer-design problem has an optimal solution.
\end{theorem}

\begin{proof}
The objective $J(W)=U(W)-\lambda L(W)-\mu C(W)$ is continuous on the compact set $\cW$. By the Weierstrass theorem, it attains a maximum.
\end{proof}

\begin{remark}
Compactness and continuity give the baseline mathematical conditions under which a proposed security filter is a well-posed observer design rather than only a procedure.
\end{remark}

\subsection{A convex leakage-control case}

Suppose the defender chooses a channel $W(y\mid x)$. Let $P_{S,X}$ be fixed and set
\[
    P_{S,Y}(s,y)=\sum_x P_{S,X}(s,x)W(y\mid x).
\]
A common leakage measure is mutual information $I(S;Y)$. Consider the constrained problem
\[
\begin{aligned}
    \text{maximize}\quad & U(W)\\
    \text{subject to}\quad & I(S;Y)\le \ell,\\
    & W(\cdot\mid x)\in\DeltaS(Y)\quad\forall x.
\end{aligned}
\]

\begin{theorem}[KKT form for finite leakage design]
Assume $U$ is differentiable and concave in $W$, and the feasible set has a strictly feasible point with $I(S;Y)<\ell$. If $W^\star$ is an interior optimum, then there exists $\lambda\ge0$ and row-normalization multipliers $\nu_x$ such that
\[
    \frac{\partial U}{\partial W(y\mid x)}(W^\star)-\lambda\frac{\partial I(S;Y)}{\partial W(y\mid x)}(W^\star)-\nu_x=0
\]
for every $x,y$, with complementary slackness
\[
    \lambda(I(S;Y)-\ell)=0.
\]
For positive joint probabilities, the derivative of the leakage term is
\[
    \frac{\partial I(S;Y)}{\partial W(y\mid x)}
    =\sum_s P_{S,X}(s,x)\log\frac{P_{S,Y}(s,y)}{P_S(s)P_Y(y)}.
\]
\end{theorem}

\begin{proof}
The feasible region is described by simplex constraints and the leakage constraint. Under strict feasibility, the standard KKT conditions apply for concave maximization with convex inequality written as $I(S;Y)-\ell\le0$. The derivative follows by differentiating
\[
I(S;Y)=\sum_{s,y}P_{S,Y}(s,y)\log\frac{P_{S,Y}(s,y)}{P_S(s)P_Y(y)}
\]
with respect to $W(y\mid x)$, using $\partial P_{S,Y}(s,y)/\partial W(y\mid x)=P_{S,X}(s,x)$ and cancellation of the terms arising from $P_Y(y)=\sum_sP_{S,Y}(s,y)$.
\end{proof}

\subsection{Dynamic observer control}

The observer may be selected over time. Let $s_t$ denote an information state, $w_t$ an observer/sensor action, and $a_t$ an adversarial action. A robust observer-control recursion can be written as
\[
    V_t(s)=\max_{w\in\cW(s)}\min_{a\in\cA(s)}\left\{r(s,w,a)-\lambda \ell(s,w,a)+\E[V_{t+1}(s')\mid s,w,a]\right\},
\]
with terminal value $V_T$. This equation is not presented as a new dynamic-programming theorem; it is the natural control representation of the security game. The reward measures monitoring utility, the leakage term measures adversarial information, and the minimization encodes worst-case attack.

\begin{proposition}[Monotone value under observer enrichment]
If the defender's available observer actions are enlarged from $\cW(s)$ to $\cW'(s)\supseteq\cW(s)$ at every state, and the leakage penalty is already included in the objective, then the robust value cannot decrease.
\end{proposition}

\begin{proof}
At each Bellman step, the maximization is taken over a superset. The value of a supremum over a larger feasible set is at least the value over the original set. Backward induction proves the proposition.
\end{proof}

\begin{remark}
This proposition is sometimes misunderstood. Enlarging the observer class available to the defender may improve the defender's optimized value. Enlarging the observer class available to the adversary refines the adversary quotient and may decrease security. The direction depends on who controls the observation.
\end{remark}

\section{Boundary Functionals on Observer Quotients}

The boundary-functional layer is stated with the hypotheses needed for stochastic approximation rather than hiding them in prose. The probability space $(\Omega,\cF,\Prb)$ is complete; $(\cF_k)$ is the observation filtration generated by public transcripts, side observations, controller memory, and environment scheduling decisions; all controls and observer choices are predictable with respect to this filtration; all boundary surrogates are nonnegative, integrable, and adapted; the state process is either compact-valued or remains almost surely in a proper sublevel set of the boundary gap; and conditional expectations are taken with respect to the declared filtration. These assumptions are exactly what allow the scalar boundary functional to be composed with cryptographic scheduling and projected control.

The quotient calculus above identifies the observer class in which a security statement is interpreted. The remaining question is how a defender establishes movement toward a secure quotient set without reading the secret state, the exact leakage objective, or the raw continuation variable. This section adds that boundary-functional layer. The construction places the boundary-observed limit-seeking argument inside the security game: the analytical object is a scalar boundary process, not a raw-state oracle.

\begin{definition}[Secure quotient manifold]
Let $q_\Pi:X\to X/{\eqobs}$ be the quotient map induced by an observer class $\Pi$. A set $\cM_\Pi\subseteq X/{\eqobs}$ is a secure quotient manifold when reaching it means that the relevant adversarial view, leakage statistic, or protocol transcript has entered the admissible security region. Its raw lift is
\[
    M_\Pi:=q_\Pi^{-1}(\cM_\Pi)\subseteq X.
\]
For example, $\cM_\Pi$ may be the set of quotient classes whose distinguishing advantage is at most a target level, or the set of monitor states satisfying a leakage budget.
\end{definition}

\begin{definition}[Separating quotient boundary gap]
A continuous map $B_\Pi:X\to\R_+$ is a separating quotient boundary gap for $\cM_\Pi$ when
\[
    B_\Pi(x)=0 \quad\Longleftrightarrow\quad q_\Pi(x)\in\cM_\Pi,
\]
and, for every $\eps>0$,
\[
    b_\eps:=\inf\{B_\Pi(x):\operatorname{dist}(q_\Pi(x),\cM_\Pi)\ge \eps\}>0.
\]
If, on a tube around $M_\Pi$, constants $0<c\le C$ and $r\ge 1$ satisfy
\[
    c\,\operatorname{dist}(q_\Pi(x),\cM_\Pi)^r\le B_\Pi(x)\le C\,\operatorname{dist}(q_\Pi(x),\cM_\Pi)^r,
\]
then $B_\Pi$ is coercive of order $r$.
\end{definition}

\begin{definition}[Oracle-free quotient functional]
Let $\cF_k^y$ be the observation filtration generated by transcripts, sensors, side-channel measurements, and internal monitor state. A nonnegative process $(z_k)$ is an oracle-free quotient surrogate when $z_k$ is $\cF_k^y$-adapted, while the feedback and stopping logic do not read any of
\[
    \operatorname{dist}(q_\Pi(x_k),\cM_\Pi),\qquad \ind\{q_\Pi(x_k)\in\cM_\Pi\},\qquad L(x_k),
\]
where $L$ denotes an exact leakage or objective value used only in the proof. The surrogate is boundary-consistent when $|z_k-B_\Pi(x_k)|\to0$ almost surely.
\end{definition}

\begin{definition}[Boundary dissipativity]
Let $\varphi:\R_+\to\R_+$ be continuous and positive definite. The closed loop is quotient-boundary dissipative when
\[
    \E[z_{k+1}\mid \cF_k^y]\le z_k-a_k\varphi(z_k)+\chi_k,
    \qquad a_k\ge0,\quad \chi_k\ge0,
\]
with $\sum_k a_k=\infty$ and $\sum_k\E\chi_k<\infty$.
\end{definition}

The proof is quotient-level. Raw dynamics may move in directions that are invisible to the admitted observer, while the theorem is stated over the quotient trace and the scalar boundary process. The boundary gate is the bounded quantity through which leakage, reliability, feasibility, and residual noise enter the distinguishing bound.

\begin{lemma}[Persistent drift forces decay to zero]
Let $(z_k)$ be nonnegative, integrable, and adapted, and suppose
\[
    \E[z_{k+1}\mid\cF_k^y]\le z_k-a_k\varphi(z_k)+\chi_k
\]
with $\varphi$ continuous and positive definite, $\sum_k a_k=\infty$, and $\sum_k\chi_k<\infty$ almost surely. Then $z_k\to0$ almost surely.
\end{lemma}

\begin{proof}
This is the Robbins--Siegmund almost-supermartingale argument in the present notation. The inequality gives convergence of $z_k$ to some finite $z_\infty\ge0$ and summability of $\sum_k a_k\varphi(z_k)$. If $z_\infty>0$ on an event of positive probability, then along that event $z_k$ is eventually bounded below by a positive level. Since $\varphi$ is positive on every compact interval away from zero, $\varphi(z_k)$ is eventually bounded below by a positive constant. The divergence of $\sum_k a_k$ then contradicts summability of $\sum_k a_k\varphi(z_k)$. Hence $z_\infty=0$ almost surely.
\end{proof}

\begin{theorem}[Boundary-bounded quotient security]
Assume $B_\Pi$ is a separating quotient boundary gap and $(z_k)$ is a boundary-consistent oracle-free quotient surrogate satisfying boundary dissipativity. Then
\[
    z_k\to0\quad\text{and}\quad \operatorname{dist}(q_\Pi(x_k),\cM_\Pi)\to0
\]
almost surely. The conclusion uses no raw-state equality oracle, no exact target-membership oracle, and no exact leakage-objective oracle.
\end{theorem}

\begin{proof}
The previous lemma gives $z_k\to0$. Boundary consistency gives $B_\Pi(x_k)\to0$. If the quotient distance did not converge to zero, then for some $\eps>0$ a subsequence would satisfy $\operatorname{dist}(q_\Pi(x_{k_j}),\cM_\Pi)\ge\eps$. By separation, $B_\Pi(x_{k_j})\ge b_\eps>0$, contradicting $B_\Pi(x_k)\to0$.
\end{proof}

\begin{proposition}[Cryptographic advantage as a boundary consequence]
Suppose a distinguishing advantage functional satisfies, for constants $L_A\ge0$ and residuals $\rho_k\ge0$,
\[
    \operatorname{Adv}^{\Pi}_{\cA}(x_k)\le L_A B_\Pi(x_k)+\rho_k
\]
for every adversary $\cA$ in the admissible class. Under the hypotheses of the boundary-bounded quotient theorem,
\[
    \limsup_{k\to\infty}\operatorname{Adv}^{\Pi}_{\cA}(x_k)\le \limsup_{k\to\infty}\rho_k.
\]
If $\rho_k\to0$, the adversary advantage vanishes along the bounded quotient trajectory.
\end{proposition}

\begin{proof}
Since $B_\Pi(x_k)\to0$ almost surely, the displayed upper bound gives the stated limsup. The proof separates surrogate convergence from the cryptographic reduction that upper-bounds advantage by the boundary gap.
\end{proof}

\begin{remark}[What was missing before]
The earlier quotient sections show that an observer can fail to see a continuation. The boundary layer gives a bounded statement over time by forcing an adapted scalar quotient gap to vanish. This is the bridge between the cryptographic advantage bound and the dissipativity estimate.
\end{remark}

\section{Rates, Residual Leakage, and Augmented Gates}

The almost-sure theorem is qualitative. A monitor with persistent noise or a side channel with nonzero floor requires a quantitative residual statement. The following rates express that residual as a leakage floor.

\begin{definition}[Boundary PL inequality]
The quotient boundary gap satisfies a boundary Polyak--Lojasiewicz inequality with modulus $\mu>0$ and order one when
\[
    \varphi(r)\ge \mu r\qquad\text{for all }r\ge0.
\]
More generally, order $p\ge1$ means $\varphi(r)\ge\mu r^p$.
\end{definition}

\begin{theorem}[Boundary rates]
Assume boundary dissipativity and the linear boundary PL inequality $\varphi(r)\ge\mu r$.
\begin{enumerate}[leftmargin=2em]
    \item If $a_k\equiv a$ with $0<a\mu<1$ and $\chi_k\le\bar\chi$, then
    \[
        \E z_k\le (1-a\mu)^k\E z_0+\frac{\bar\chi}{a\mu}.
    \]
    \item If $a_k=c/(k+1)$ and $\E\chi_k=O(k^{-(1+s)})$ for $s>0$, then
    \[
        \E z_k=O\bigl(k^{-\min(c\mu,s)}\bigr),
    \]
    with the usual additional logarithm when $c\mu=s$.
\end{enumerate}
If $B_\Pi$ is coercive of order $r$, the same envelopes transfer to $\E[\operatorname{dist}(q_\Pi(x_k),\cM_\Pi)^r]$ up to the coercivity constant and the surrogate-consistency error.
\end{theorem}

\begin{proof}
Taking total expectations in the dissipativity inequality and using $\E\varphi(z_k)\ge \mu\E z_k$ gives a scalar comparison recursion. Constant gain yields the affine contraction and its geometric series. The diminishing-gain statement is the standard Chung comparison for $m_{k+1}\le(1-c\mu/(k+1))m_k+O(k^{-(1+s)})$. Coercivity gives $\operatorname{dist}(q_\Pi(x_k),\cM_\Pi)^r\le B_\Pi(x_k)/c$.
\end{proof}

\begin{proposition}[Input-to-state leakage floor]
Assume $a_k\ge a>0$, $\chi_k\le\bar\chi$, and $\varphi$ is continuous, strictly increasing, positive definite, and convex. Define
\[
    r^\star=\varphi^{-1}(\bar\chi/a).
\]
Then the deterministic envelope $m_k=\E z_k$ satisfies
\[
    \limsup_{k\to\infty}m_k\le r^\star.
\]
If $B_\Pi$ is coercive of order $r$, then
\[
    \limsup_{k\to\infty}\E\operatorname{dist}(q_\Pi(x_k),\cM_\Pi)^r
    \le \frac{1}{c}\,\varphi^{-1}(\bar\chi/a).
\]
\end{proposition}

\begin{proof}
Jensen's inequality gives $\E\varphi(z_k)\ge\varphi(\E z_k)$, hence
\[
    m_{k+1}\le m_k-a\varphi(m_k)+\bar\chi.
\]
Above $r^\star$, the right-hand side has negative drift. The comparison recursion therefore cannot have an asymptotic envelope above $r^\star$. The distance statement again follows by coercivity.
\end{proof}

\begin{figure}[!htbp]
\centering
\includegraphics[width=0.72\linewidth]{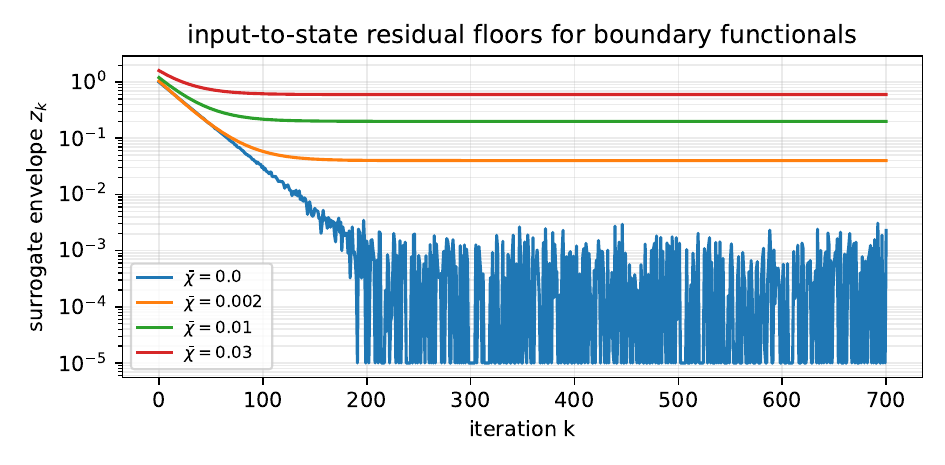}
\caption{Persistent residuals enter as leakage floors. The scalar comparison recursion gives the input-to-state behavior: summable error collapses to zero, bounded error leaves a quantified floor.}
\label{fig:residual-floor}
\end{figure}
\FloatBarrier

\begin{definition}[Augmented security gate]
Let
\[
    \xi_k=(x_k,z_k,r_k,c_k,\ell_k),
\]
where $z_k$ is the quotient boundary surrogate, $r_k\in[0,1]$ is reliability, $c_k$ is resource or feasibility slack, and $\ell_k$ records recurrence or repeated terminal rejections. For weights $\lambda,\mu,\nu\ge0$, define
\[
    G_k(\xi_k)=z_k+\lambda(1-r_k)^+ +\mu(-c_k)^+ +\nu R(\ell_k),\qquad R\ge0.
\]
\end{definition}

\begin{proposition}[Augmented gate reduction]
If $G_k$ is adapted and satisfies
\[
    \E[G_{k+1}\mid\cF_k^y]\le G_k-a_k\varphi(G_k)+\chi_k,
\]
with $\sum_k a_k=\infty$ and $\sum_k\chi_k<\infty$, then $G_k\to0$ almost surely. In particular, if $0\le z_k\le G_k$ and $z_k$ is boundary-consistent, the quotient trajectory approaches $\cM_\Pi$.
\end{proposition}

\begin{proof}
Apply the persistent-drift lemma to $G_k$. Since every summand in $G_k$ is nonnegative, $G_k\to0$ forces the leakage surrogate, reliability deficit, resource violation, and recurrence penalty to vanish together.
\end{proof}

\section{Oracle-Free Sensor Control and Safety Filters}

The boundary surrogate can be used as a control objective. A defender may choose sensors, telemetry, padding, throttling, masking, or probing actions. The action is admissible only if it respects resource constraints and does not create a new unsafe channel. The resulting synthesis problem is a constrained control problem over observer channels.

\begin{definition}[Preferred boundary action]
A preferred action $u_k^{\rm pref}$ is boundary-admissible when it is $\cF_k^y$-measurable and is chosen to decrease the boundary gate using only observed quantities. It may use randomized excitation, but it may not read $\operatorname{dist}(q_\Pi(x_k),\cM_\Pi)$, target membership, or an exact leakage objective.
\end{definition}

\begin{proposition}[One-step realization]
Suppose that for each $k$ the controller can choose an admissible action such that the next gate satisfies
\[
    \E[G_{k+1}\mid\cF_k^y]\le G_k-a_k\varphi(G_k)+\chi_k.
\]
Then the closed loop satisfies the augmented gate theorem, hence the quotient boundary functional converges to zero under summable residual.
\end{proposition}

\begin{proof}
The displayed one-step inequality is exactly the dissipativity hypothesis for $G_k$. The previous proposition applies.
\end{proof}

For continuous-time realization, let
\[
    \dot x=f(x)+g(x)u
\]
and let $b(x,t)\ge0$ define the admissible safe set. A time-varying control-barrier condition has the form
\[
    L_f b(x,t)+L_g b(x,t)u+\partial_t b(x,t)\ge -\alpha_b(b(x,t)).
\]
Capability constraints are written as affine inequalities
\[
    c_j(x,t)+d_j(x,t)^\top u\ge0,
    \qquad j\in J.
\]
The realized action is the projection
\[
    u^\circ(x,t)=\argmin_{u}\;\frac12\|u-u^{\rm pref}(x,t)\|^2
\]
subject to the barrier and capability inequalities.

\begin{theorem}[Safety-compatible quotient surrogate]
Assume the feasible set of the projection above is nonempty with a uniform Slater margin on compact subsets. Then the projected controller is well-defined and preserves the safe set $\{b\ge0\}$. If the boundary-gate dynamics under $u^{\rm pref}$ satisfy $\dot G\le-\varphi(G)+\epsilon(t)$ and the projection correction $\Delta=u^\circ-u^{\rm pref}$ enters as
\[
    \dot G\le-\varphi(G)+\epsilon(t)+L_G\|\Delta(t)\|,
\]
then:
\begin{enumerate}[leftmargin=2em]
    \item if $\epsilon\in L^1$ and $\int_0^\infty\|\Delta(t)\|dt<\infty$, then $G(t)\to0$;
    \item if $\|\epsilon\|_\infty\le\bar\epsilon$ and $\|\Delta\|_\infty\le\bar\Delta$, then
    \[
        \limsup_{t\to\infty}G(t)\le \varphi^{-1}(\bar\epsilon+L_G\bar\Delta).
    \]
\end{enumerate}
Thus an interior safety margin gives exact convergence of the surrogate, while a boundary-active safety filter gives a quantified residual.
\end{theorem}

\begin{proof}
The projection is a Euclidean projection onto a nonempty closed convex polyhedron, hence exists and is unique. Since $u^\circ$ satisfies the barrier inequality, the standard comparison argument for $\dot b\ge-\alpha_b(b)$ keeps $b$ nonnegative. The gate estimates follow by integrating the displayed differential inequality. In the $L^1$ case, Barbalat's lemma gives $G(t)\to0$ under uniform continuity. In the bounded-input case, the scalar comparison equation yields the stated asymptotic gain.
\end{proof}

\subsection{When projection preserves descent and when it creates a floor}

Let $\tilde x_{k+1}$ be the boundary-descending nominal update and let $x_{k+1}=\tilde x_{k+1}+\Delta_k$ be the projected safe update. If the quotient gap $B$ is $L_B$-Lipschitz on the safe tube, then
\[
B(x_{k+1})\le B(\tilde x_{k+1})+L_B\|\Delta_k\|.
\]
Consequently a nominal drift inequality becomes
\[
\E[Z_{k+1}\mid\cF_k]\le Z_k-a_k\phi(Z_k)+\chi_k+L_B\E[\|\Delta_k\|\mid\cF_k].
\]
The projection is harmless when $\sum_k\E\|\Delta_k\|<\infty$; it creates an ISS floor when $\limsup_k\E\|\Delta_k\|\le \bar\Delta$.

\begin{theorem}[Projected-surrogate dichotomy]
Assume the measurability, compactness, and stability hypotheses of the boundary theorem, and assume the safe projection is feasible at every step. If $\sum_k\E\|\Delta_k\|<\infty$, then the projected loop has the same almost-sure quotient convergence as the nominal loop. If $a_k\equiv a>0$, $\phi(r)\ge\mu r$, and $\limsup_k\E\|\Delta_k\|\le\bar\Delta$, then
\[
\limsup_k \E Z_k\le \frac{\bar\chi+L_B\bar\Delta}{a\mu}.
\]
Thus projection either disappears into the summable residual or becomes a measurable leakage-control floor.
\end{theorem}

\begin{proof}
Insert the Lipschitz correction into the scalar drift. The summable case is Robbins--Siegmund with enlarged summable residual. The persistent case is the affine ISS recursion with residual $\bar\chi+L_B\bar\Delta$.
\end{proof}

\section{Orientation-Free Excitation and Hidden-Manifold Example}

A random probe may help discover a boundary, but its sign is not the functional. The relevant invariant is the scalar dissipativity inequality. This distinction matters in security because a probe direction can be attacker-controllable, biased, or meaningless under quotient collapse.

\begin{definition}[Orientation-free excitation]
An adapted perturbation $\eta_k$ is orientation-free when
\[
    \E[\eta_k\mid\cF_k^y]=0,
    \qquad
    \E[\eta_k\eta_k^\top\mid\cF_k^y]=\sigma_k^2 I.
\]
\end{definition}

\begin{lemma}[Directional cancellation]
For every square-integrable $\cF_k^y$-measurable vector $g_k$,
\[
    \E[\langle g_k,\eta_k\rangle\mid\cF_k^y]=0.
\]
\end{lemma}

\begin{proof}
By conditional linearity,
$\E[\langle g_k,\eta_k\rangle\mid\cF_k^y]=\langle g_k,\E[\eta_k\mid\cF_k^y]\rangle=0$.
\end{proof}

\begin{lemma}[Second-order radialization]
Let $h:\R^d\to\R$ be twice continuously differentiable near $x$. For small $\sigma>0$ and an orientation-free perturbation $\eta$ with finite third moment,
\[
    \E[h(x+\sigma\eta)-h(x)\mid\cF^y]
    =\frac{\sigma^2}{2}\operatorname{tr}\!\left(\nabla^2h(x)\E[\eta\eta^\top\mid\cF^y]\right)+O(\sigma^3).
\]
In particular, balanced perturbation contributes no signed first-order boundary invariant; it reveals curvature only at second order.
\end{lemma}

\begin{proof}
Taylor expansion gives $h(x+\sigma\eta)-h(x)=\sigma\langle\nabla h(x),\eta\rangle+\frac{\sigma^2}{2}\eta^\top\nabla^2h(x)\eta+O(\sigma^3\|\eta\|^3)$. Conditional expectation kills the first-order term by directional cancellation and converts the quadratic term to the trace expression.
\end{proof}

In a hidden-manifold instance, the raw state can circulate around a target boundary while the admissible observable is only a scalar residual. The relevant assertion is therefore the decay of the observed boundary statistic to its residual floor; a trajectory plot is secondary to the dissipativity inequality that follows.

\begin{example}[Boundary-observed side-channel hardening]
Let a device have internal state $x_k$ and an observer quotient induced by timing, power, and transcript channels. The secure quotient manifold $\cM_\Pi$ is the set of quotient classes whose leakage is below a budget. A randomized padding or scheduling action supplies orientation-free excitation; a monitor reads only $z_k$, an estimated boundary residual. If the defender proves
\[
    \E[z_{k+1}\mid\cF_k^y]\le z_k-a_k\varphi(z_k)+\chi_k,
\]
then the leakage surrogate converges. If the monitor has a persistent noise floor, the ISS proposition reports the remaining adversary advantage as a floor. The argument asserts convergence to the declared quotient boundary rather than equality of internal representatives.
\end{example}

\section{Finite-Model Proof Obligations}

For a finite observer-quotient game, a proof is determined by the raw experiments, the admitted observer family, the quotient relation, the transition kernel, the adversary class, and the boundary-to-advantage map. Changing any one of these objects changes the theorem. The transition kernel is compatible with the quotient or carries an explicit one-step defect; the environment's post-processing remains inside the closure of the observer family; and the scalar boundary functional upper-bounds the distinguishing advantage of every scheduled observer.

\begin{proposition}[Finite audit bound]
Let $X$ be finite and let $\sigma_\Pi(x)$ be the observer signature of $x$. Suppose that, for all states with the same signature,
\[
    \TV\bigl(q_{\Pi\#}K(x,\cdot),q_{\Pi\#}K(x',\cdot)\bigr)\le \delta_K,
    \qquad
    |B_\Pi(x)-B_\Pi(x')|\le \delta_B,
\]
and that every admitted post-processing map has representation defect at most $\delta_{post}$. For any adaptive environment running for $T$ continuation steps,
\[
    \Adv_T\le T(\delta_K+\delta_{post})+\rho\!\bigl(\E z_T+T\delta_B\bigr)+\delta_{sim}(T),
\]
where $z_T$ is the terminal boundary surrogate and $\rho$ is the declared boundary-to-advantage map.
\end{proposition}

\begin{proof}
Condition on each transcript prefix and expose the next observer selected by the scheduler. The finite signature condition gives a per-step quotient-pushforward defect $\delta_K$, post-processing contributes $\delta_{post}$, and the boundary mismatch contributes $\delta_B$ through the monotone map $\rho$. A telescoping hybrid over the $T$ continuations gives the sum; the simulator-interface hop contributes $\delta_{sim}(T)$.
\end{proof}

The same obligation appears in the linear control reading. The observer signature is the finite-horizon output map, quotient compatibility is equality of block transition vectors, and a hidden continuation is a vector in the output kernel. A proof that adds a sensor row therefore records the old kernel, the refined kernel, the resulting boundary pseudometric, and the change in the residual advantage floor.

\section{Case Studies}

\subsection{Hidden counter with delayed leakage}

Consider a protocol that increments an internal counter $n$ on every request. The functional output is independent of $n$ until a maintenance packet reveals $n\bmod 2$. A one-step test that excludes maintenance packets sees no leakage. In our notation, the observer class $\Pi_0$ contains only ordinary outputs, so $(s,n)\eqobs(s,n+1)$ under $\Pi_0$. When the maintenance observer is added, the quotient refines and the continuation is no longer hidden. The delayed leakage counterexample in Theorem~5.2 is exactly this structure.

\subsection{Side-channel-aware encryption}

Let the raw state contain $(k,r,m,c,\tau)$: key, randomness, message, ciphertext, and timing state. A transcript-only observer sees $c$. A side-channel observer sees $(c,\tau)$. A proof under the first observer class establishes only $d_{\Pi_{ct}}(x_0,x_1)\approx0$. It says nothing about $d_{\Pi_{ct+time}}(x_0,x_1)$. The refinement monotonicity proposition explains the failure: $\Pi_{ct}\subseteq\Pi_{ct+time}$, so the second quotient is finer.

\subsection{Telemetry design in a controlled plant}

Let a plant have state $x=(x_s,x_h)$, where $x_s$ is safety-relevant and $x_h$ is hidden under the current sensor $C$. If $x_h$ can later drive $x_s$, then ignoring $x_h$ is unsafe even if current outputs match. The design problem is to add a sensor row $c$ that maximizes detection utility subject to bandwidth and privacy constraints. This is precisely the finite observer optimization problem of Section~8.

\section{Computational Observer Classes}

The preceding sections used statistical distance because it makes the quotient geometry explicit. Cryptographic security usually restricts the distinguisher. We now state the computational version without changing the underlying quotient language.

\begin{definition}[Resource-bounded observer pseudometric]
Let $\cA_t$ be a class of distinguishers with resource bound $t$, for example circuit size, running time, number of oracle queries, or a mixed physical-computational budget. For two states $x,x'$ and observer class $\Pi$, define
\[
    d_{\Pi,t}(x,x')=\sup_{\pi\in\Pi}\sup_{A\in\cA_t}
    \left|\Prb[A(\pi(x))=1]-\Prb[A(\pi(x'))=1]\right|.
\]
Write $x\equiv_{\Pi,t,\eps}x'$ when $d_{\Pi,t}(x,x')\le\eps$.
\end{definition}

The statistical quotient $X/\Pi$ is replaced by a resource-indexed approximate quotient. Increasing the resource bound refines the quotient. This is the computational analogue of observer refinement.

\begin{proposition}[Resource monotonicity]
If $\cA_t\subseteq\cA_{t'}$ and $\Pi\subseteq\Pi'$, then
\[
    d_{\Pi,t}(x,x')\le d_{\Pi',t'}(x,x')
\]
for all states $x,x'$.
\end{proposition}

\begin{proof}
Both suprema are taken over subsets of the larger feasible sets.
\end{proof}

\begin{definition}[Negligible observer security]
For a security parameter $\kappa$, a family of experiment pairs is computationally observer-quotient secure if for every polynomial $p$ there exists a negligible function $\nu$ such that
\[
    \operatorname{Adv}_{\Pi_\kappa,T(\kappa)}^{\cA_{p(\kappa)}}(\mathsf{Exp}_{0,\kappa},\mathsf{Exp}_{1,\kappa})
    \le \nu(\kappa).
\]
\end{definition}

\begin{proposition}[Standard indistinguishability as a special case]
Let $\Pi_{ct}$ contain only the ciphertext/transcript observer of an encryption experiment, and let $\cA$ be the usual polynomial-time chosen-plaintext adversary class. Then observer-quotient security of the left-or-right experiments is exactly IND-style security for that transcript model.
\end{proposition}

\begin{proof}
The two experiments are the standard left-or-right worlds. The observer emits precisely the transcript delivered to the adversary. Therefore the OQSG advantage is the standard distinguishing advantage, with notation changed but no mathematical content lost.
\end{proof}

\begin{remark}
The proposition records the embedding of the standard transcript model. The quotient formulation becomes informative when the observer class changes, for example from ciphertext-only to ciphertext-plus-timing or ciphertext-plus-cache-state.
\end{remark}

\section{Lattice and Order Structure of Observers}

Observer classes can be compared. This gives a small algebra that is useful when a proof is moved from one threat model to another.

\begin{definition}[Observer preorder]
For two observer classes $\Pi_1,\Pi_2$, write $\Pi_1\preceq \Pi_2$ when every distinction made by $\Pi_1$ is also made by $\Pi_2$; equivalently,
\[
    x\equiv_{\Pi_2}x' \implies x\equiv_{\Pi_1}x'
\]
for all $x,x'$.
\end{definition}

If $\Pi_1\subseteq\Pi_2$, then $\Pi_1\preceq\Pi_2$, but the preorder is slightly more general because post-processings do not add information.

\begin{theorem}[Union refines, intersection coarsens]
Let $\Pi_1,\Pi_2$ be observer classes. Then
\[
    x\equiv_{\Pi_1\cup\Pi_2}x'
    \quad\Longleftrightarrow\quad
    x\equiv_{\Pi_1}x'\ \text{and}\ x\equiv_{\Pi_2}x'.
\]
Thus the quotient induced by $\Pi_1\cup\Pi_2$ is the common refinement of the two quotients.
\end{theorem}

\begin{proof}
The left side means every observer in the union gives equal output distributions. This is equivalent to equality for every observer in $\Pi_1$ and every observer in $\Pi_2$ separately.
\end{proof}

\begin{theorem}[Post-processing closure leaves the quotient unchanged]
Let $\operatorname{Post}(\Pi)$ be the class of all observers obtained by composing an element of $\Pi$ with an arbitrary Markov post-processing. Then
\[
    x\equiv_{\Pi}x'\quad\Longleftrightarrow\quad x\equiv_{\operatorname{Post}(\Pi)}x'.
\]
\end{theorem}

\begin{proof}
Since $\Pi\subseteq\operatorname{Post}(\Pi)$, the right side implies the left. Conversely, if $x\equiv_{\Pi}x'$, then every $\pi\in\Pi$ has equal output law on $x,x'$. Applying any post-processing kernel to equal laws keeps them equal.
\end{proof}

\begin{definition}[Observer-stable property]
A predicate $P$ on states is observer-stable under $\Pi$ if $x\eqobs x'$ and $P(x)$ imply $P(x')$.
\end{definition}

\begin{proposition}[Quotient factorization of stable predicates]
A predicate $P$ is observer-stable under $\Pi$ if and only if there exists a predicate $\bar P$ on $X/\Pi$ such that $P=\bar P\circ q_{\Pi}$.
\end{proposition}

\begin{proof}
If $P=\bar P\circ q_{\Pi}$, then equal quotient states have equal predicate value. Conversely, if $P$ is stable, define $\bar P([x])=P(x)$. Stability makes this well-defined.
\end{proof}

Observer stability means that a property has the same truth value throughout each observer block. If the truth value varies inside a block, the observer lacks the resolution needed for that property.

\section{Fixed Points on Quotient Lattices}

A transition may not have a raw fixed point, but the quotient transition may. This is the mathematical core of observational fixedness.

\begin{definition}[Quotient fixed point]
Let $T:X\to X$ be deterministic. A quotient class $Q\in X/\Pi$ is a quotient fixed point when
\[
    q_{\Pi}(T(x))=q_{\Pi}(x)
\]
for one, and hence every quotient-compatible representative, $x\in Q$.
\end{definition}

\begin{proposition}[Raw fixed point implies quotient fixed point]
If $T(x)=x$, then $q_{\Pi}(x)$ is fixed by the induced quotient dynamics whenever that dynamics is well-defined.
\end{proposition}

\begin{proof}
Apply $q_{\Pi}$ to both sides of $T(x)=x$.
\end{proof}

The converse is false exactly when hidden continuations exist. The next result gives a useful existence condition for quotient fixed points even when raw fixed points are absent.

\begin{theorem}[Finite quotient recurrence gives quotient fixedness]
Let $T:X\to X$ be deterministic and quotient-compatible. If the quotient $X/\Pi$ is finite, then every quotient trajectory eventually enters a periodic orbit. If the period is one, the reached class is a quotient fixed point. More generally, the periodic orbit is an observational cycle invisible to any observer that samples only orbit-invariant statistics.
\end{theorem}

\begin{proof}
The induced map $\bar T:X/\Pi\to X/\Pi$ is a function on a finite set. Every forward orbit of a function on a finite set eventually repeats. Once a state repeats, the deterministic orbit is periodic. Period one is exactly a fixed point.
\end{proof}

\begin{theorem}[Monotone observer closure fixed point]
Let $\mathfrak{P}$ be the complete lattice of partitions of a finite state space ordered by refinement. Let $F:\mathfrak{P}\to\mathfrak{P}$ be a monotone operator that maps a candidate observer partition to the partition obtained after one round of transition-and-observation refinement. Then $F$ has a least and greatest fixed point.
\end{theorem}

\begin{proof}
This is an immediate application of Tarski's fixed point theorem to the complete lattice $\mathfrak{P}$.
\end{proof}

\begin{remark}
In model checking terms, the least fixed point corresponds to the coarsest stable quotient compatible with the chosen observations and transition structure. In security terms, it is the weakest quotient in which local observer equivalence can be propagated safely through the dynamics.
\end{remark}

\section{Algorithmic Checking for Finite Systems}
\label{sec:finitechecking}

For finite systems, the definitions above are constructive. This section gives the algorithmic core that can be implemented without symbolic mathematics.

\subsection{Finite observer signatures}

Assume $X$ is finite, $|X|=n$, and $\Pi=\{\pi_1,\ldots,\pi_m\}$ is finite. The exact observer signature of a state is
\[
    \sigma(x)=\big(\pi_1(\cdot\mid x),\ldots,\pi_m(\cdot\mid x)\big).
\]
Then $x\eqobs x'$ if and only if $\sigma(x)=\sigma(x')$.

\begin{proposition}[Signature construction]
For finite deterministic observers, the observer quotient can be computed by sorting states by their signatures. If all observer outputs are represented as exact symbols, this takes $O(mn\log n)$ comparisons up to the cost of comparing signatures.
\end{proposition}

\begin{proof}
The signature is a complete invariant of observer equivalence by definition. Sorting groups equal signatures into equivalence classes.
\end{proof}

\subsection{Compatibility checking}

For a Markov kernel $K$ and an observer partition $\mathcal{P}=\{B_1,\ldots,B_r\}$, define the block transition vector
\[
    \beta_K(x)=\left(K(x)(B_1),\ldots,K(x)(B_r)\right).
\]
The kernel is quotient-compatible exactly when $\beta_K(x)=\beta_K(x')$ for all $x,x'$ in the same block.

\begin{proposition}[Finite compatibility test]
Given a partition $\mathcal{P}$ of $X$, a finite kernel $K$ is compatible with $\mathcal{P}$ if and only if all states in each block have the same block transition vector. The test costs $O(nr)$ arithmetic operations after block probabilities have been computed.
\end{proposition}

\begin{proof}
Compatibility requires equal pushforward distributions on quotient blocks. The block transition vector is exactly that pushforward distribution.
\end{proof}

\subsection{Refinement algorithm}

If the initial observer quotient is not transition-compatible, it can be refined until it is. The following procedure is the observer-quotient analogue of partition refinement.

\begin{enumerate}[leftmargin=2em,label=\textnormal{\arabic*.}]
    \item Start with the partition $\mathcal{P}_0$ induced by observer signatures.
    \item At iteration $i$, compute the block transition vector $\beta_K^i(x)$ of every state with respect to $\mathcal{P}_i$.
    \item Refine each block of $\mathcal{P}_i$ by grouping states with equal pairs $(\sigma(x),\beta_K^i(x))$.
    \item Stop when no block changes and return the fixed partition $\mathcal{P}_\star$.
\end{enumerate}

\begin{theorem}[Correctness of stable quotient refinement]
For finite $X$, the stable observer quotient algorithm terminates after at most $n-1$ strict refinements and returns the coarsest refinement of the initial observer quotient with respect to which $K$ is compatible.
\end{theorem}

\begin{proof}
Each strict refinement increases the number of blocks by at least one, and there are at most $n$ singleton blocks, so termination occurs after at most $n-1$ strict refinements. At termination, all states in each block have equal block transition vectors, hence compatibility holds. The algorithm only splits states when compatibility forces a split; by induction, every compatible refinement of the initial observer quotient also splits every pair split by the algorithm. Therefore the returned partition is the coarsest compatible refinement.
\end{proof}

\begin{remark}
This gives a practical way to find delayed leakage. If a pair of states is equivalent in $\mathcal{P}_0$ but separated in $\mathcal{P}_\star$, then the current observer does not distinguish them immediately, but the dynamics can move them into distinguishable future blocks.
\end{remark}

\section{Structural Failure Modes}

The theory points to concrete proof obligations. The following rules are phrased as diagnostics rather than theorems.

\subsection{Unspecified observer class}

A security theorem without a declared observer class has an implicit observer class. That hidden assumption is unstable. A theorem under transcript-only observation is not a theorem under transcript-plus-timing observation.

\subsection{One-step equality without continuation compatibility}

A local equality $\pi(x)=\pi(x')$ controls the present observation. Trace security additionally requires quotient compatibility of future transitions or an explicit defect bound.

\subsection{Unmodeled side-channel refinement}

Adding a side channel changes the quotient. In the notation of this paper, it replaces $\Pi$ by $\Pi\cup\Pi_{side}$. A proof depending on $X/\Pi$ is then interpreted over the refined quotient.

\subsection{Implicit observer design}

When a system intentionally releases telemetry, logs, proofs, or public statistics, the release mechanism is treated as a designed channel $W$. The objective state utility, leakage, and cost. This makes the design reproducible and comparable across threat models.

\section{Finite-State Reproducibility Artifact}

The source package contains a finite-state artifact under \texttt{anc/}. The artifact implements the equations used in the side-channel and LTI sections and records generated data, CSV ledgers, and JSON summaries. The computation is synthetic: its role is to verify that the quotient operations in the paper are executable and that the numerical outputs match the stated formulas.

\subsection{Gaussian side-observer benchmark}

Let $x,x'$ be two representatives of the same public transcript quotient class, so $d_{\Pi_{tr}}(x,x')=0$. A profiled leakage observer returns
\[
    Y_x=\alpha\operatorname{HW}(S_x)+N,
    \qquad N\sim\mathcal N(0,\sigma^2).
\]
For means $\mu_x,\mu_{x'}$, define
\[
    B_{prof}(x,x')=\frac{(\mu_x-\mu_{x'})^2}{2\sigma^2}.
\]
With $n$ independent leakage samples, the optimal equal-variance Gaussian distinguishing advantage is
\[
    \Adv_n(x,x')=2\Phi\!\left(\sqrt{\frac{nB_{prof}(x,x')}{2}}\right)-1.
\]
For $m$ tested pairs, the simultaneous lower bound for the mean displacement is
\[
    \underline\Delta(x,x')=
    \left(|\widehat\mu_x-\widehat\mu_{x'}|-2\sigma\sqrt{\frac{2\log(4m/\delta)}{n}}\right)_+,
\]
and $\underline B_{prof}=\underline\Delta^2/(2\sigma^2)$. A pair is split only when $\underline B_{prof}>r^\star$.

The accompanying script generates eight transcript-equivalent state pairs under a matched continuation schedule with $n=240$ samples per representative, $\sigma=1$, $\alpha=0.28$, and seed $1729$. The generated ledgers report
\[
\begin{aligned}
    \max d_{\Pi_{tr}} &= 0,\\
    \operatorname{median} B_{prof} &= 0.3920,\\
    \operatorname{median}\widehat{\Adv}_n^{\rm lower} &= 0.9535,\\
    \#\{\text{pairs split after simultaneous correction}\} &= 5/8.
\end{aligned}
\]
These values instantiate the finite-state refinement theorem: transcript equivalence can coexist with a profiled side-observer separation, and the partition is refined only through the simultaneous lower-confidence rule.

\subsection{Observer-row design audit}

The LTI audit computes the observability matrix, the information matrix, the kernel dimension, and the residual advantage bound for a small sensor library. For the public sensor alone, the audit records a nontrivial hidden direction. Adding the hidden-coordinate sensor removes that direction and lowers the residual floor. The script also evaluates a relaxed cost frontier using the residual model
\[
    r(C)=\frac{\bar\chi}{\lambda_{\min}(G_T(C))+\bar\chi},
    \qquad \rho(r)=\sqrt r\wedge 1.
\]
The generated log reports the old and refined kernel dimensions, the residual floors, and the resulting reduction in the bounded advantage term. This matches the theorem in Section~\ref{sec:lti-example}: increasing the weakest observed direction reduces the residual term that enters $\rho$.

\subsection{Ancillary layout}

The arXiv source archive places ancillary material under \texttt{anc/}. The two scripts are stored there: \texttt{synthetic\_observer\_quotient\_benchmark.py} and \texttt{lti\_observer\_redesign\_audit.py}. The script \texttt{run\_all.sh} regenerates the logs. Synthetic samples are stored in \texttt{anc/data/}; CSV and JSON ledgers are stored in \texttt{anc/logs/}. The TeX source and figure PDFs remain at the archive root.

\section{Conclusion}

Observer-quotient security games give a single language for cryptographic indistinguishability, side-channel leakage, and sensor-based control. The central distinction is between raw state equality and observer equality. Hidden continuations live in the gap. They are harmless when the quotient dynamics is compatible and trace lifting applies; they are dangerous when a later observation can reveal what was previously hidden. The same formalism also supports observer synthesis: an implementation can trade utility, leakage, measurement cost, and residual floor inside a single composable advantage bound.

\appendix

\section{Additional Proof Details}

\subsection{Approximate observer equivalence is stable under post-processing}

\begin{lemma}
If $x\eqep x'$ and $R$ is any post-processing kernel, then $R\circ\pi$ differs by at most $\eps$ on $x,x'$ for every $\pi\in\Pi$.
\end{lemma}

\begin{proof}
By the post-processing theorem,
\[
\TV((R\circ\pi)(\cdot\mid x),(R\circ\pi)(\cdot\mid x'))\le\TV(\pi(\cdot\mid x),\pi(\cdot\mid x'))\le\eps.
\]
\end{proof}

\subsection{Blackwell order interpretation}

One observer channel $W_1$ is less informative than another observer channel $W_2$ if there exists a post-processing kernel $R$ such that $W_1=R\circ W_2$. In that case, every distinction available under $W_1$ is already available under $W_2$. The observer quotient induced by $W_2$ refines the quotient induced by $W_1$. This is the decision-theoretic form of the post-processing theorem.

\end{document}